\documentclass[a4paper, 12pt]{article}
\usepackage{fullpage}
\usepackage[utf8]{inputenc}
\usepackage[T1]{fontenc}
\usepackage[english]{babel}

\usepackage{graphicx}
\usepackage{amsmath}
\usepackage{amsthm}
\usepackage{mathrsfs}
\usepackage{amssymb}
\usepackage{cancel}
\usepackage[ruled,vlined]{algorithm2e}

\usepackage{tikz}
\usepackage{subcaption}
\usetikzlibrary{positioning}
\usetikzlibrary{decorations.pathreplacing}
\usetikzlibrary{plotmarks,trees,arrows}
\usetikzlibrary{calc}
\usetikzlibrary{graphs}
\usetikzlibrary{shapes.geometric}
\usepackage{pgf}
\usepackage{color}

\newtheorem{theorem}{Theorem}
\newtheorem{lemma}[theorem]{Lemma}
\newtheorem{corollary}[theorem]{Corollary}

\newtheorem{remark}[theorem]{Remark}

\tikzset{
    mynode/.style={draw,circle,fill=\c!30,scale=\s},
    root/.style={draw,rectangle,fill=\c!30,scale=\s}
}

\newcommand{\monparagraph}[1]{\smallskip\noindent{\bf #1}\ }

\newcommand{\Neig}{\mathcal{N}}
\newcommand{\Pa}{\mathcal{P}}
\newcommand{\Ch}{\mathcal{C}}
\newcommand{\Si}{\mathcal{S}}
\newcommand{\bz}{\mathbf{0}}
\newcommand{\bo}{\mathbf{1}}
\newcommand{\bk}{\mathbf{k}}
\newcommand{\Ss}{\texttt{SS}^1}

\newcommand{\Bb}{\texttt{B}}
\newcommand{\rk}{k}

\newcommand{\Rer}{$\mathscr{R}_{er}$}
\newcommand{\Rreset}{$\mathscr{R}_{reset}$}
\newcommand{\RerRank}{$\mathscr{R}_{erRank}$}
\newcommand{\Rjoin}{$\mathscr{R}_{join}$}
\newcommand{\Rtok}{$\mathscr{R}_{tok}$}
\newcommand{\Radd}{$\mathscr{R}_{add}$}
\newcommand{\Rready}{$\mathscr{R}_{ready}$}
\newcommand{\RerRoot}{$\mathscr{R}^r_{er}$}
\newcommand{\Rgo}{$\mathscr{R}^r_{go}$}
\title{Deterministic Synchronous Self-Stabilizing BFS Construction\\
       with Constant Space Complexity}
\author{L\'elia Blin \and Franck Petit \and S\'ebastien Tixeuil}
\date{}

\begin{document}
\maketitle
\begin{abstract}

In this paper, we resolve a long-standing open problem in self-stabilization asking whether it is possible to construct a spanning tree using \emph{constant} memory per node in synchronous semi-uniform networks, i.e., networks in which one node is distinguished.
We design a synchronous self-stabilizing  algorithm that constructs a breadth-first search (BFS) tree in any \emph{anonymous} semi-uniform network using only a constant number of bits of memory per node. 
Crucially, our approach operates without any prior knowledge of global network parameters such as maximum degree, diameter, or number of nodes. In contrast to traditional self-stabilizing methods — such as pointer-to-neighbors, distance-to-root, or identifiers — that are unsuitable under strict memory constraints, our solution employs an innovative constant-space token dissemination mechanism. This mechanism effectively eliminates cycles and rectifies errors in the BFS structure, ensuring both correctness and memory efficiency. The proposed algorithm not only meets the stringent requirements of memory-constrained distributed systems, but also opens new avenues for research in the design of self-stabilizing protocols under severe resource limitations: constant space-complexity may not systematically prevent the existence of self-stabilizing algorithms for important non-trivial tasks.
\end{abstract}

\section{Introduction}

This paper addresses the challenge of developing memory-efficient self-stabilizing algorithms for the Breadth-First Search (BFS) Spanning Tree Construction problem.
Self-stabilization~\cite{ADDP19,D74,D00,H22} is a versatile technique that facilitates recovery after arbitrary \textit{transient} faults impact the distributed system, affecting both the participating processes and the communication medium. Essentially, a self-stabilizing protocol can restore the system to a legitimate configuration (from which its behavior satisfies its specification), beginning from an arbitrary, potentially corrupted, initial global state, without the need of any human intervention.
A BFS spanning tree construction involves \textit{(i)} each node identifying a single neighbor as its parent, with the exception of the root node, which has no parent, and \textit{(ii)} the set of parent-child relationships forming an acyclic connected structure, ensuring that following the parent nodes from any node results in the shortest path to the root node.
Memory efficiency pertains to the amount of information transmitted to neighboring nodes to enable stabilization. A smaller space complexity results in reduced information transmission, which \textit{(i)} decreases the overhead of self-stabilization in the absence of faults or after stabilization, and \textit{(ii)} facilitates the integration of self-stabilization and replication~\cite{GCH06,HP00}.

We adopt the classical \textit{state model} introduced in~\cite{D74}.
In this model, a node can read its own memory and the memory of its neighbors, and execute its algorithm accordingly when needed. This means that a node's memory is accessible to all its neighbors, ensuring they all obtain the same information. 
Consequently, if a node $v$ wants to communicate information exclusively to its neighbor $u$, $v$ must indicate in its memory that the information is intended only for $u$. This indication is typically made by adding node $u$'s identifier.
%This communication pattern is commonly found in wireless networks, where targeting a specific neighboring node requires the transmitter to include the intended recipient's identifier in the message. 
This type of communication is inherent to wireless networks, where nodes exchange information by broadcasting messages over one or more carrier waves. In such networks, the notion of {\emph neighboring node} of $v$ refers to any other node able to directly receive messages emitted by $v$, typically without the aid of relays.  Consequently, targeting a specific neighbor necessitates that the sender explicitly includes the recipient's ID within the transmitted message. This identification requirement adds a significant overhead in memory space, specifically $O(\log n)$ bits in a network containing $n$ nodes. Within anonymous semi-uniform networks, nodes are devoid of identifiers, and only a single node assumes a role distinct from the rest. More precisely, the distinguished node operates under a dedicated protocol, whereas all other nodes execute an identical protocol. As a result of anonymity, it is not possible to rely on identifiers to interact with or refer to a specific node.

In~\cite{DGS99}, the authors introduce the notion of a \textit{silent algorithm} (defined in a model slightly different from the state model~\cite{D74}), a fundamental concept that significantly influences the space complexity of self-stabilizing algorithms.  An algorithm is said to be \textit{silent} if every execution eventually reaches a point beyond which the state of all nodes remains unchanged.
They show that in networks of $n$ nodes, solving tasks such as leader election or spanning tree construction silently requires at least $\Omega(\log n)$ bits of memory per node. 
However, in the state model, a self-stabilizing algorithm is not necessarily silent, and \textit{talkative} (\textit{i.e.}, non-silent) approaches were successful in reducing the memory cost of self-stabilizing algorithms~\cite{BT18,BT20}. 

In general networks, such as those considered in this paper, self-stabilizing leader election is closely linked to self-stabilizing tree construction. 
On one hand, the presence of a leader enables efficient self-stabilizing tree construction~\cite{BPR13,CYH91,CD94,DDJL23,DIM93,IL94,J97,KKM11}. 
On the other hand, growing and merging trees is the primary technique for designing self-stabilizing leader election algorithms in networks, as the leader often serves as the root of an inward tree~\cite{AB98,AKY90,AG94,BDPR10,BT18,BT20}.

From a memory perspective, the aforementioned works employ two algorithmic techniques for self-stabilization that negatively impact memory complexity.
The first technique involves using a \textit{distance} variable to store the distance of each node to the elected node in the network. This distance variable is employed in self-stabilizing spanning tree construction to break cycles resulting from arbitrary initial states, see for instance \cite{AB98,AKY90,AG94}.
Clearly, storing distances in $n$-node networks may require $\Omega(\log n)$ bits per node. Some works~\cite{AO94,BT18,BT20} distribute pieces of information about the distances to the leader among the nodes using different mechanisms, allowing for the storage of $o(\log n)$ bits per node. However, in general networks, the best result so far~\cite{BT20} uses $O(\log \log n + \log \Delta)$ bits per node, closely matched by the $\Omega(\log \log n)$ bits per node lower bound obtained in~\cite{BFB23}.

The second technique involves using a \textit{pointer-to-neighbor} variable to unambiguously designate a specific neighbor for each node (with the designated neighbor being aware of this). For tree construction, pointer-to-neighbor variables typically store the parent node in the constructed tree (with the parent being aware of its children). In a naive implementation, the parent of each node is designated by its identifier, requiring $\Omega(\log n)$ bits for each pointer variable. However, it is possible to reduce the memory requirement to $O(\log \Delta)$ bits per pointer variable in networks with a maximum degree of $\Delta$ by using node-coloring at distance 2 instead of identifiers to identify neighbors. 
The best available deterministic self-stabilizing distance-2 coloring protocol~\cite{BT20} achieves a memory footprint of $O(\log \log n + \log \Delta)$ bits per node, assuming unique identifiers. 
This implies that even self-stabilizing protocols solely relying on a constant number of pointer-to-neighbor variables per node~\cite{J97,DDJL23} actually use at best $O(\log \log n + \log \Delta)$ bits per node (assuming unique identifiers) in the classical state model, and that a lower bound of $\Omega(\log \Delta)$ bits per node also holds for these solutions (regardless of the nodes being anonymous or not).
 
We now turn to the results of Itkis and Levin~\cite{IL94}, which constitute one of the works most closely related to ours.  In their approach, the authors propose a randomized asynchronous self-stabilizing algorithm for leader election, based on a \emph{deterministic} construction of a BFS tree. Their algorithms are designed  under the same model as in this paper, namely the \emph{state model}, but augmented with \emph{pointer-to-neighbor variables}. As already mentioned above, using such additional assumption, each node $u$ can maintain a unique pointer to a neighbor $v\in \mathcal{N}(u)$; more importantly, this neighbor $v$ is aware that it is being pointed to by $u$.
Three deterministic BFS construction algorithms are presented by Itkis and Levin~\cite{IL94}, proving successive improvements on their space complexity. Their first algorithm is a standard BFS construction based on the distance from each node to the root, and it requires $O(\log n)$ bits of memory per node. The other two algorithms are presented in natural language, and their formalization as a set of guarded rules is left undefined. The former is restricted to graphs of degree~2. It is based on a sophisticated and elegant cycle-detection mechanism based on the Thue-Morse sequence. 
Note that although no formal proofs exist in the literature, the inherent properties of the Thue-Morse sequence and finite automata theory suggest that it is impossible to implement this sequence within finite memory. Nevertheless, the authors use a Thue-Morse-based certificate to obtain a constant space complexity per node. The second of the two algorithms generalizes the first to arbitrary graphs. This algorithm maintains a forest of pointers, and guarantees acyclicity of these pointers thanks to the use of the degree-2 algorithm. The latter is applied to each path from the root of the BFS tree to its leaves, where each path is identified thanks to a depth-first search (DFS) traversal of the BFS tree.  These DFS traversals benefit from the pointer-to-neighbor variables. To keep the space complexity constant, the DFS traversals are performed successively. This requires to keep track of which DFS traversals have already been performed. The authors argue that this can be done using constant space complexity using the pointer-to-neighbor variables allowed by their model. However, in absence of such pointers, it is not clear how to implement this sequence of DFS traversals using less than $\Omega(\log \Delta)$ bits of memory per node~\cite{CD94,DattaJPV00}.

Overall, up to this paper, there exists no deterministic self-stabilizing solution to the spanning tree construction problem that uses a constant number of bits of memory per process. Furthermore, the reuse of known techniques makes this goal unattainable.

\noindent\textbf{Our contribution.}
We introduce the first deterministic self-stabilizing Breadth-First Search Spanning Tree construction algorithm that utilizes a constant amount of memory per node in the classical state model, irrespective of network parameters such as degree, diameter, or size. Our algorithm functions in arbitrary topology networks and does not depend on any global knowledge.

In~\cite{BFB23}, it is shown that constructing a self-stabilizing spanning tree without a pre-existing leader using constant memory is impossible. Therefore, a necessary assumption for achieving constant space per node is that the anonymous network is semi-uniform, with a single node designated as the leader, serving as the root of the spanning tree, while all other nodes are anonymous and uniformly execute the same algorithm. Additionally, for synchronization purposes, our solution assumes a fully synchronous network, where all nodes operate at the same pace in a lock-step manner.

Our solution relies on a few key components. 

First, each non-root node maintains a \textit{rank} variable, which classifies neighbors as parents (if their rank is lower), siblings (if their rank is the same), or children (if their rank is higher). 
In the stabilized phase, the rank value corresponds to the distance to the root modulo three (hence, the BFS property of the constructed spanning tree), an idea already suggested in various non self-stabilizing settings~\cite{A05,BDLP08,CFIKP08}. The self-stabilizing property adds a significant challenge to the task, as the initial configuration may make the rank variable induce cycles or multiple roots, from which the algorithm must recover. 
Since identifying the single parent of each node (if it exists) cannot make use of node identifiers (as all non-root nodes are anonymous), instead we select the parent node (among the set of parents) whose port number is minimal.
However, nodes do \textit{not} explicitly communicate the specific parent they have selected, as this would necessitate $O(\log \Delta)$ bits of memory.

The second component of our algorithm uses a similar technique as the \emph{Power Supply} introduced by Afek and Bremler~\cite{AB98}.
In both cases, the root continuously floods the network with tokens that have a limited lifespan to
circumvent one of the main issues associated with token-based approaches: the infinite circulation of tokens within a
cycle.  
The token's lifespan is controlled by a binary counter, where each bit is stored on a distinct node. So, the tokens generated by the root decay exponentially (half of the tokens disappear at the root's neighbors, half of the remaining ones at the neighbors' downward neighbors, and so on). Tokens are used to destroy cycles and correct inconsistent situations that may appear in arbitrary initial configurations, observing that in synchronous networks, tokens passed deterministically arrive simultaneously at all nodes equidistant from the root (and thus at all neighbors with the same rank for any given node, if the ranks are correct). The exponential decay prevents tokens from looping indefinitely in the network, without relying on any global knowledge.

Overall, our algorithm uses only 6 bits of memory per node (actually, only 36 states), and has a stabilization time of $O(2^{\varepsilon})$ time units, where $\varepsilon$ denotes the root eccentricity (unknown to the participating nodes).

\section{Preliminaries}
\label{sec:model}

In this section, we define the underlying distributed execution model considered in this paper, and state what it means for a protocol to be self-stabilizing.  

A \textit{distributed system} is an undirected connected graph, $G=(V,E)$, where $V$ is a set of vertices (or nodes )--- $|V|=n,\ n \geq 2$ --- and $E$ is the set of edges connecting those nodes. 
The distance between two nodes $u$ and $v$, denoted by $d^u_v$ is the length of the shortest path between $u$ and $v$ in $G$. 
Note that since $G$ is not oriented, $d^u_v=d^v_u$, and if $u=v$, then $d^u_v=0$. Then, the eccentricity of a node $v$ is the maximum distance from $v$ to any other node in the network. 

Nodes represent processes, and edges represent bidirectional communication links (that is, the ability for two nodes to communicate directly).  
We assume that each node $v$ is able to distinguish each incident edge with a locally assigned unique label called \textit{port number}. 
Port numbers are immutable and are not coordinated in any way: the edge $(u,v)$ may correspond to port number $k$ at $u$ and to port number $k'\neq k$ at $v$.
Let $\Neig(v)$ be the set of port numbers, also called \textit{set of neighbors} of $v$. 
The \textit{degree} of $v$ denotes $|\Neig(v)|$. 

We consider the state model \cite{D74} in the context of semi-uniform anonymous 
networks. In this model, the communication paradigm allows each node 
to read its own memory as well as the memories of its neighbors.
In the self-stabilizing literature, the \emph{memory} (or 
\emph{space complexity}) of a protocol refers to the number of bits 
stored by each node to represent its variables; this is the quantity 
we optimize in this paper.
Let $v$ be any node in $V$, and let $\{u_1, \ldots, u_{\delta(v)}\}$ 
denote its set of neighbors, where $\delta(v)$ is the degree of the node $v$. No neighbor $u_i,\ i \in \{1, \ldots, 
\delta(v)\}$, can determine which port number in $\delta(v)$ 
corresponds to the edge $(v,u_i)$. Conversely, since nodes have no 
identifiers, $v$ cannot designate any particular $u_i$ to receive 
$v$'s information. Hence, a pointer-to-neighbor cannot be implemented 
solely by reading the neighbors' variables. Note that the port numbers 
assigned by $v$ to its neighbors may be used when evaluating the rules, 
but are not stored as variables and thus do not contribute to the 
memory. They are instead part of the \emph{computation memory}, which 
denotes the working space used by $v$ to evaluate its rules given the 
variables of $\{u_1, \ldots, u_{\delta(v)}\}$; this quantity is 
bounded by $\Delta \cdot s$, where $s$ is the memory of the 
protocol and $\Delta$ the degree of the given graph.

A semi-uniform anonymous network is characterized by the absence of node identifiers and the presence of a unique distinguished node. This node serves as the root of the BFS tree and follows a distinct protocol, while all other nodes adhere to a uniform protocol.
The protocol of a node consists of a set of variables and a set of (guarded) rules, of the following form: 
$$<label>:\ <guard>\ \longrightarrow <statement>$$
Each node may write its own variables, and read its own variables and those of its neighboring nodes.
The guard of a rule is a predicate on the variables of $v$ and its neighbors. The statement of a rule of $v$ updates one or more variables of $v$. 
A rule may be executed only if its guard evaluates to true.   
The rules are atomically executed, so the evaluation of a guard and the execution of the corresponding statement, if any, are done in one atomic step.

The \textit{state} of a node is defined by the values of its variables.
The \textit{configuration} of a system is the product of the states of all nodes. 
Let $\Gamma$ be the set of all possible configurations of the system.
A distributed algorithm (or, protocol) $\mathcal{A}$ is a collection of binary transition
relations, denoted by $\mapsto$, on $\Gamma$ such that given two configurations $\gamma_1$ and $\gamma_2$, 
$\gamma_1 \mapsto \gamma_2$ by the atomic execution of guarded action of one or more nodes.  
A protocol  $\mathcal{A}$ induces an oriented graph $\Gamma^{\mapsto} =(\Gamma, \mapsto)$, called the \textit{transition graph} of $\mathcal{A}$.
A sequence $e=\gamma_0, \gamma_1, \ldots, \gamma_i,\gamma_{i+1},\ldots$, $\forall i\geq 0, \gamma_i \in \Gamma$, is called an \textit{execution} of $\mathcal{A}$ if and only if $\forall i\geq 0, \gamma_{i}\mapsto \gamma _{i+1} $. 

A node $v$ is said to be \textit{enabled} in a configuration $\gamma \in \Gamma$ if there exists a rule $R$ such that 
the guard of $R$ is true at $v$ in $\gamma$.
In this paper, we assume that each transition of $\Gamma^{\mapsto}$ is driven by a \textit{synchronous scheduler}. 
This means that, given a pair of configurations $\gamma_1, \gamma_2 \in \Gamma$, then $\gamma_1\mapsto \gamma_2 \in \Gamma^{\mapsto}$ if and only if all enabled nodes in $\gamma_1$ execute an atomic rule during the transition $\gamma_1 \mapsto \gamma_2$. 

\monparagraph{Self-Stabilization.}
A predicate $P$ is \textit{closed} for a transition graph $\Gamma^{\mapsto}$ if and only if every configuration in any execution $e$ starting from a configuration satisfying $P$ also satisfies $P$.  
A predicate $Q$ is an \textit{attractor} of $P$, denoted $P \vartriangleright Q$, if and only if:  
\textit{(i)} $Q$ is closed for $\Gamma^{\mapsto}$, and \textit{(ii)} for every execution $e$ on $\Gamma^{\mapsto}$ beginning in a configuration satisfying $P$, $e$ contains at least one configuration satisfying $Q$.  
A transition graph $\Gamma^{\mapsto}$ is \textit{self-stabilizing} for a predicate $P$ when $P$ is an attractor of the predicate \textit{true} (formally, $\mathit{true} \vartriangleright P$), ensuring convergence to $P$ from any initial configuration.  

\monparagraph{Problem to be solved.}

We assume a \textit{rooted graph}, meaning that the set of vertices includes a specific node called the \textit{root}, denoted by $r$.
Let us call $V^r$ the set $V$ with the extra property that each node $v \in V$ is labeled with $d^r_v$. 
For every node $v \neq r$, let $Pred_v$ be the
set of neighboring nodes of $v$ in $V^r$ that are labeled with $d^r_v - 1$.  
Let $E^r$ be the subset of edges ($E^r \subset E$) such that for every edge $uv \in E^r$, either $u$ belongs to $Pred_v$, or $v$ belongs to $Pred_u$.
So, in $G^r=(V^r,E^r)$, no two neighbors have the same label. 
Then $\overrightarrow{G^r}$ is the directed version of $G^r$ every edge is oriented toward the smaller label. 
By construction, $\overrightarrow{G^r}$ is a Directed Acyclic Graph (DAG) rooted in $r$ such that every decreasing label path from every $v \in V^r$ is a path of minimum length between $v$ and $r$. Note that this DAG is uniquely defined. In the remainder of the paper, we refer to $\overrightarrow{G^r}$ as $BFS\mbox{-}DAG$ (rooted at $r$).  

Let $BFS^{\equiv 3}$ be the projection of the $BFS\mbox{-}DAG$, where each vertex $v$ is labeled with $d^r_v \mod 3$.
The problem we consider in this paper consists in the design of a self-stabilizing distributed algorithm that builds the $BFS^{\equiv 3}$ using a constant amount of memory. From this construction, one can retrieve the BFS tree by choosing for each node (except the root) the parent port number of minimum value; as noted above, this selection is 
computed on the fly and does not contribute to the space complexity. 

\section{Algorithm}
\subsection{Overview of the Algorithm}
\label{subsec:overview}

Our algorithm $\mathtt{TokBin}$ is composed of two sub-algorithms 
running concurrently: $\mathtt{TokBinR}$ (Algorithm~\ref{algo:TokBinR}), 
executed solely by the root~$r$, and $\mathtt{TokBinV}$ 
(Algorithm~\ref{algo:TokBinV}), executed by every node 
$v \in V \setminus \{r\}$. Together, they ensure that each node $v$ 
eventually computes its distance to $r$ modulo~$3$.

%\paragraph{Token flow.}
$\mathtt{TokBinR}$ continuously generates tokens at the root and 
propagates them downward through the network. Once a node acquires a 
token, it offers it to its children in the next step; if tokens arrive 
simultaneously from multiple parents, they are merged into a single 
token. Since the model is synchronous, all parents of $v$ at the same 
distance from the root send their token in the same step; node $v$ 
therefore receives them simultaneously and merges them into a single 
token before forwarding it to its own children. Since the model is synchronous, all tokens generated by 
the root at a given step progress together, preserving an identical 
distance from the root at every step. A node $v$ updates its rank — 
its distance to $r$ modulo $3$ — based on the rank of the neighbors 
from which it receives a token: $v$ adopts the rank of its token 
senders incremented by one modulo $3$.

%\paragraph{Limited token lifespan.}
Since the initial configuration is arbitrary, two sources of 
inconsistency may prevent a valid $BFS^{\equiv 3}$ from forming: tokens 
not generated by the root may already be present in the network, and 
the rank variables may induce cycles in which tokens circulate 
indefinitely.

To address both issues, each node maintains a bit variable that limits 
how far a token can travel. Each time a node receives a token, it 
flips its bit; a token is only forwarded if the bit was set to $\bo$ 
at the moment of reception, otherwise the token disappears. This 
mechanism ensures that every token — whether generated by the root or 
present in the initial configuration — has a bounded lifespan: it can 
travel at most a finite number of hops before being dropped.

For tokens present in the initial configuration, this lifespan is 
short enough that they all disappear within $O(n)$ steps, since no 
cycle in the network can sustain a token indefinitely — after at most 
$n$ flips along any cycle, the bit pattern forces the token to be 
dropped.

For cycles induced by incorrect rank variables, the limited lifespan 
plays a different role. Once all initial tokens have disappeared, 
every cycle in the network is token-free. Since the graph is 
connected, at least one node $v$ in the cycle has a neighbor outside 
the cycle that will eventually offer it a token. This token comes from 
a node whose state is inconsistent with other neighbors of $v$, so $v$ 
detects an error and resets itself, breaking the cycle.

Globally, this mechanism can be understood as maintaining a binary 
word along each path from the root to any node $v$. Tokens travel 
along these paths and perform binary addition, one increment at a 
time. A token reaches $v$ if and only if the binary word along the 
path from the root to $v$ has all bits set to $\bo$ — meaning the 
entire path is simultaneously ready to forward. This is what gives 
$\mathtt{TokBin}$ its $O(2^\varepsilon)$ stabilization time, where 
$\varepsilon$ is the eccentricity of the root.
To illustrate this mechanism, consider a chain $r, v_1, v_2, \dots, v_\ell$ where all bits are 
initially set to $\bz$, giving the binary word $\bz\bz\bz\dots\bz$. 
When the root offers a token, $v_1$ acquires it but does not forward 
it; instead, $v_1$ flips its bit. The binary word becomes 
$\bo\bz\bz\dots\bz$. At the next token offer, $v_1$ forwards the 
token to $v_2$, flips back, and $v_2$ flips: the word becomes 
$\bz\bo\bz\dots\bz$. This process continues, effectively incrementing 
the binary counter by one at each token emission. For the farthest 
node $v_\ell$ to receive a token, the root must emit $2^\ell$ tokens.

%\paragraph{Avoiding deadlock.}
This token-based approach also circumvents a classical issue in 
self-stabilizing spanning tree construction. The tree is built from 
the root outward: the root never waits for information coming back 
from the leaves, which avoids deadlock when the network contains only 
cycles and no leaves exist.

%\paragraph{Non-silence.}
Although $\mathtt{TokBin}$ eventually stabilizes — meaning the rank 
of each node eventually equals its distance to $r$ modulo $3$ and no 
longer changes — the algorithm is not silent. Even after stabilization, 
the token-passing mechanism keeps operating indefinitely. 

\subsection{Local Variables}

Each node $v$ maintains three variables, summarized in Table~\ref{tab:variables}.

\paragraph*{Rank variable denoted by $\rk$}
The rank variable stores the distance from $v$ to the root modulo~$3$. Upon stabilization, $\rk_v \in \{0,1,2\}$ induces the kinship structure matching $BFS^{\equiv 3}$: a neighbor $u$ of $v$ is a parent if $\rk_u = (\rk_v - 1) \bmod 3$, a child if $\rk_u = (\rk_v + 1) \bmod 3$, and a sibling if $\rk_u = \rk_v$. The extra value $\bot$ indicates that $v$ has been \emph{reset} — its rank is unknown — typically following an error. In an arbitrary initial configuration, $\rk_v$ may be incorrect, inducing cycles or DAGs rooted at nodes other than $r$; the token mechanism described below is responsible for correcting these inconsistencies.

\paragraph*{Token variable denoted by $t$}
The token variable drives the diffusion mechanism that stabilizes $\rk_v$. It takes values in $\{\mathit{false}, \mathit{true}, \mathit{wait}\}$. The states $\mathit{true}$ and $\mathit{false}$ indicate whether $v$ currently holds a token or not. The intermediate state $\mathit{wait}$ serves two purposes: it prevents $v$ from immediately receiving a new token after having just acquired one, and it signals to neighboring nodes that $v$ has released its token in the previous step — a key distinction exploited by the error-detection predicates (Section~\ref{subsub:ED}). After stabilization, every node cycles through $\mathit{false} \to \mathit{true} \to \mathit{wait} \to \mathit{false}$, indefinitely; a node may remain in the $\mathit{false}$ state for more than one step while waiting for a token to reach it.

\paragraph*{Bit variable denoted by $b$.}
The bit variable implements the binary counter that limits token lifespan, as described in Section~\ref{subsec:overview}. Its normal states $\{\bz, \bo\}$ encode one bit of the distributed binary word formed along each path from the root. Each time a node receives a token, it flips its bit, which conditionally suppresses token forwarding and realizes the exponential decay. The special value $\top$ is reserved for error signaling: when $v$ enters an error state, it sets $b_v := \top$ so that its neighbors can detect and propagate the error (see $Er_{prg}$ in Section~\ref{subsub:ED}).

\paragraph*{Root's fixed point.}
All three variables are also maintained at the root $r$, which however executes Algorithm $\mathtt{TokBinR}$ rather than Algorithm $\mathtt{TokBinV}$. Under normal operation, the root holds the fixed values $\rk_r = 0^*$, $t_r = \mathit{true}$, and $b_r = \bo$, continuously proposing a token to all its neighbors. The root deviates from this fixed point only transiently, to signal an error (Section~\ref{subsec:TokBinR}).

%\paragraph*{Summary}
\begin{table}[h]
\centering
\renewcommand{\arraystretch}{1.4}
\begin{footnotesize}
\begin{tabular}{c l l l l}
\hline
\textbf{Var.} & \textbf{Role} & \textbf{Domain} & \textbf{Stabilized state} & \textbf{Special value} \\
\hline
$\rk$ & Distance to $r$ mod $3$ & $\{0,1,2,\bot\}$ & $d^r_v \bmod 3$ & $\bot$: node is reset \\
$t$   & Token holding / passing & $\{\mathit{false},\mathit{true},\mathit{wait}\}$ & cycles $\mathit{false}\!\to\!\mathit{true}\!\to\!\mathit{wait}$ & — \\
$b$   & Binary counter bit / error flag & $\{\bz,\bo,\top\}$ & alternates $\bz \leftrightarrow \bo$ & $\top$: error state \\
\hline
\multicolumn{5}{l}{\small Root $r$: fixed at $\rk_r=0^*,\ t_r=\mathit{true},\ b_r=1$ under normal operation.}\\
\hline
\end{tabular}
\end{footnotesize}
\caption{Summary of the three variables maintained by each node $v \in V$.}
\label{tab:variables}
\end{table}

%---

\subsection{Family relationships sets and Reset state}
\label{subsec:family}

In the following, we denote by $\Neig(v)$ the set of neighbors of node $v$. The parent, child, and sibling relationships are defined by the rank variable $\rk$ of each node. Note that $\Pa(v)$, $\Ch(v)$, and $\Si(v)$ are only defined when $\rk_v \neq \bot$; when $\rk_v = \bot$, we set $\Pa(v) = \Ch(v) = \Si(v) = \emptyset$ by convention.

The set of \emph{parents} of $v$ consists of all neighbors whose rank is exactly one less than $\rk_v$ modulo~$3$:
$$\Pa(v)=\{u\in \Neig(v)\mid \rk_u=(\rk_v-1)\bmod 3\}.$$

Symmetrically, the set of \emph{children} of $v$ consists of all neighbors whose rank is one greater:
$$\mathcal{C}(v)= \{u\in \Neig(v)\mid \rk_u=(\rk_v+1)\bmod 3\}.$$

Finally, the set of \emph{siblings} of $v$ consists of all neighbors sharing the same rank as $v$:
$$\mathcal{S}(v)= \{u\in \Neig(v)\mid \rk_u=\rk_v\}.$$

During the BFS construction, some nodes may detect an error and disconnect themselves from the current spanning structure. To signal this disconnection, a node sets its variables to specific default values, entering a special state called the \emph{reset state}. A reset node does not belong to any family relationship: it has no parents, no children, and no siblings. It remains in this dormant state, waiting to be reintegrated into the spanning structure. Formally, a node $v$ is in the reset state if:
$$Reset(v)\equiv \rk_v=\bot \wedge t_v=\mathit{false} \wedge b_v=\bz.$$

\subsection{Algorithm for the root : $\mathtt{TokBinR}$}
\label{subsec:TokBinR}
Although the root $r$ does not participate in the spanning construction as a regular node, it still requires its own algorithm to ensure the correct behavior of its neighbors. To reason about the state of its neighborhood, we introduce three predicates. Predicate $\mathcal{N}_{Synch}(r)$ captures the normal operating condition of the root's neighborhood: all neighbors hold rank $\bo$ and share identical values for variables $b$ and $t$, reflecting the fact that, in a correct execution, all neighbors of $r$ are at distance $1$ from the root and progress through the token-passing steps simultaneously.
\begin{equation}
    \mathcal{N}_{Synch}(r)\equiv  \forall u,v\in \Neig(r)\mid \rk_u=\rk_v=1\wedge b_u=b_v\wedge t_u=t_v
\end{equation}

Predicate $\mathcal{N}_{Reset}(r)$ holds when all neighbors are in the reset state, meaning they have all disconnected from the spanning structure and are waiting to be reintegrated. 
\begin{equation}
    \mathcal{N}_{Reset}(r)\equiv \forall u\in \Neig(r)\mid Reset(u)
\end{equation}
Predicate $\mathcal{N}_{Er}(r)$ holds when all neighbors are already in the error state.
\begin{equation}
    \mathcal{N}_{Er}(r)\equiv \forall u\in \Neig(r)\mid b_u=\top
\end{equation}

These three predicates drive the two rules of Algorithm $\mathtt{TokBinR}$. The role of rule \RerRoot\/ is to detect any inconsistency in the neighborhood of $r$: if the neighborhood is neither synchronized, nor fully reset, nor fully in error, then an inconsistency exists that must be corrected. In this case, the root signals the error to its neighbors by setting $b_r := \top$, which triggers their own error detection mechanisms. Rule \Rgo\/ governs the recovery of the root: once all neighbors have successfully reset themselves, i.e., $N_{Er}(r)$ holds, the root restores its normal operating state by setting $b_r := \bo$, thereby resuming the emission of tokens toward its neighbors.
\begin{algorithm*}[!h]
\[
\begin{array}{lcl}
\mathscr{R}^r_{er} &:& b_r\neq\top \wedge \neg(\mathcal{N}_{Synch}(r)\vee \mathcal{N}_{Reset}(r)\vee \mathcal{N}_{Er}(r))\longrightarrow b_r:=\top\\[4pt]
\mathscr{R}^r_{go} &:& b_r=\top \wedge \mathcal{N}_{Er}(r) \longrightarrow b_r:=\mathbf{1}
\end{array}
\]
\caption{Algorithm $\mathtt{TokBinR}$ for node $r$}
\label{algo:TokBinR}
\end{algorithm*}

\subsection{Token Reception Predicates}
\label{subsec:token-predicates}

Before presenting Algorithm~$\mathtt{TokBinV}$, we introduce the predicates that govern
how a node $v$ may receive a token from its neighbors. These predicates are central to
three rules of $\mathtt{TokBinV}$ (\Rtok, \RerRank, and \/\Rjoin) and depend on the
parent-child structure defined in Section~\ref{subsec:family}.

\paragraph*{Senders set: $\mathit{Tok}(v)$.}
A neighbor $u$ of $v$ may offer a token to $v$ only if $u$ currently holds a token
(\emph{i.e.}, $t_u = \mathit{true}$) \emph{and} its bit satisfies $b_u = \bo$. The condition
$b_u = \bo$ is what enforces the limited token lifespan described in
Section~\ref{subsec:overview}: a node whose bit is $\bz$ suppresses token forwarding,
effectively realizing one step of the binary counter. The set of candidate senders of
$v$ is therefore:
\begin{equation}
    \mathit{Tok}(v) = \bigl\{u \in \mathcal{N}(v) \mid t_u = \mathit{true} \wedge b_u = \bo\bigr\}
    \label{eq:Tok}
\end{equation}

\paragraph*{Normal token reception: $\mathit{TakeP}(v)$.}
When $v$ holds a valid rank ($\rk_v \neq \bot$), it can identify its parent set $\Pa(v)$.
The predicate $\mathit{TakeP}(v)$ captures the ideal situation: every parent of $v$
simultaneously offers a token, and no other neighbor does. This exact match is required
because the network is synchronous — all parents at the same distance from the root must
send their tokens in the same step. If $\mathit{TakeP}(v)$ holds, $v$ executes rule
\Rtok\/ to acquire the token (merging all parents' tokens into one).
\begin{equation}
    \mathit{TakeP}(v) \equiv \mathit{Tok}(v) \neq \emptyset
                          \wedge \mathit{Tok}(v) = \Pa(v)
    \label{eq:TakeP}
\end{equation}

\paragraph*{Anomalous token reception: $\mathit{TakeO}(v)$.}
The predicate $\mathit{TakeO}(v)$ is the complement of $\mathit{TakeP}(v)$ for ranked
nodes: tokens arrive, but either from a node that is not a parent of $v$, or not from all
parents simultaneously. Both cases indicate a structural inconsistency — the current rank
assignment is incompatible with the token wave. When $\mathit{TakeO}(v)$ holds, $v$
executes rule \RerRank\/ to declare itself in error.
\begin{equation}
    \mathit{TakeO}(v) \equiv \mathit{Tok}(v) \neq \emptyset
                           \wedge \mathit{Tok}(v) \neq \Pa(v)
    \label{eq:TakeO}
\end{equation}

\paragraph*{Token reception for reset nodes: $\mathit{TakeR}(v)$.}
A reset node ($\rk_v = \bot$) has no parent set to compare against. Instead, it waits
for a \emph{consistent} wave of tokens: all senders must share the same rank, and no
other neighbor may hold that same rank. This double condition ensures that $v$ is offered
tokens from a single, unambiguous level of the BFS-DAG, with no competing candidate at
the same depth. To formalize this, we first define $\mathit{Tok}_{rk}(v)$ as the rank of
the token senders:
\begin{equation}
    \mathit{Tok}_{rk}(v) = \min\bigl\{\rk_u \mid u \in \mathit{Tok}(v)\bigr\}
    \label{eq:TokRk}
\end{equation}

Note that the $\min$ is used merely to identify the rank of the senders in a well-defined way; if senders held differing ranks, the equality in $\mathit{TakeR}(v)$ would be violated, so when $\mathit{TakeR}(v)$ holds, all senders necessarily share the same rank.
Intuitively, $\mathit{Tok}_{rk}(v)$ identifies the unique rank level from which $v$ 
should be receiving tokens; predicate $\mathit{TakeR}(v)$ then verifies that this 
level accounts for \emph{all} neighbors of that rank, with no sender missing and 
no outsider included.
\begin{equation}
   \mathit{TakeR}(v) \equiv \mathit{Tok}(v) \neq \emptyset
    \;\wedge\; \mathit{Tok}(v) = \bigl\{u \in \Neig(v) \mid \rk_u = \mathit{Tok}_{rk}(v)\bigr\}
\end{equation}
When $\mathit{TakeR}(v)$ holds, $v$ executes rule \Rjoin\/ to enter the spanning
structure: it sets its rank to $\rk_t(v)$ and acquires a token, where
\begin{equation}
    \rk_t(v) = \bigl(\mathit{Tok}_{rk}(v) + 1\bigr) \bmod 3.
\end{equation}

\paragraph*{Summary}
Table~\ref{tab:take} summarizes the three predicates and the rule each triggers.

\begin{table}[h]
\centering
\renewcommand{\arraystretch}{1.4}
\begin{tabular}{c l l l}
\hline
\textbf{Predicate} & \textbf{Applies to} & \textbf{Meaning} & \textbf{Rule triggered} \\
\hline
$\mathit{TakeP}(v)$ & ranked node ($\rk_v \neq \bot$)
    & tokens from exactly all parents
    & \Rtok (accept token) \\
$\mathit{TakeO}(v)$ & ranked node ($\rk_v \neq \bot$)
    & tokens inconsistent with parent set
    & \RerRank (enter error) \\
$\mathit{TakeR}(v)$ & reset node ($\rk_v = \bot$)
    & consistent wave from one rank level
    & \Rjoin\/ (join structure) \\
\hline
\end{tabular}
\caption{Token reception predicates and their associated rules in $\mathtt{TokBin}$.}
\label{tab:take}
\end{table}

\subsection{Non-Root Algorithm: $\mathtt{TokBinV}$}

\begin{algorithm*}[!htbp]
\noindent
\[
\begin{array}{lcl}
\mathscr{R}_{er}&:&  b_v\neq\top \wedge Er(v)  \longrightarrow t_v:=\mathit{wait};b_v:=\top\\\\
\mathscr{R}_{reset}&:&   b_v=\top \vee Er_{var}(v) \longrightarrow  \rk_v:=\bot;t_v:=\mathit{false};b_v:=\bz\\\\
\mathscr{R}_{erRank}&:& \neg Er(v) \wedge \rk_v\neq\bot \wedge b_v\neq\top \wedge t_v=\mathit{false} \wedge  \mathit{TakeO}(v) \longrightarrow t_v:=\mathit{wait};b_v:=\top\\\\
\mathscr{R}_{join}&:& \neg Er(v) \wedge \rk_v=\bot \wedge  \mathit{TakeR}(v) \longrightarrow \rk_v:=\rk_t(v);b_v:=\bo \\\\
\hrulefill\\
\mathscr{R}_{tok}&:& \neg Er(v) \wedge \rk_v\neq\bot \wedge t_v=\mathit{false}  \wedge \mathit{TakeP}(v)   \longrightarrow t_v:=\mathit{true} \\\\
\mathscr{R}_{add}&:& \neg Er(v) \wedge t_v=\mathit{true} \longrightarrow t_v:=\mathit{wait};b_v:=(b_v+1)\mod 2  \\\\
\mathscr{R}_{ready}&:& \neg Er(v) \wedge t_v=\mathit{wait} \longrightarrow t_v:=\mathit{false} \\\\
\end{array}
\]
\caption{Algorithm $\mathtt{TokBinV}$ for Node~$v\in V \setminus \{r\}$}
\label{algo:TokBinV}
\end{algorithm*}

\begin{figure}[http!]
\centering
\begin{tikzpicture}[
  every node/.style={font=\small},
  state/.style={draw=none},
  arr/.style={->, >=stealth, thick},
  redarr/.style={arr, color=red!70!black},
  pinkarr/.style={arr,dotted},
  bluearr/.style={arr,dashed},
]

\node[state] (A) at (0,0)   {$\rk^i, \mathit{true}, b^i$};
\node[state] (B) at (3.5,0) {$\rk^i, \mathit{wait}, \bar{b}^i$};
\node[state] (C) at (7,0)   {$\rk^i, \mathit{false}, b^i$};
\node[state] (E) at (3.5,2.8) {$\rk^i, \mathit{wait}, \top$};
\node[state] (R) at (7,2.8) {$\bot, \mathit{false}, \mathbf{0}$};

\draw[arr] (A) -- node[above]{\footnotesize$\mathcal{R}_{add}$} (B);
\draw[arr] (B) -- node[above]{\footnotesize$\mathcal{R}_{ready}$} (C);
\draw[arr] (C) to[out=-30, in=-150] 
  node[below]{\footnotesize$\mathcal{R}_{tok}$} (A);

\draw[pinkarr] (A) -- 
   (E);
\draw[pinkarr] (B) -- node{\footnotesize \Rer $\vee$ \RerRank}(E);
\draw[pinkarr] (C) -- (E);

\draw[bluearr] (E) -- node[above]{\footnotesize$\mathcal{R}_{reset}$} (R);

\draw[bluearr] (R) to 
  node[right]{\footnotesize$\mathcal{R}_{join}$} 
  (C);

\end{tikzpicture}
\caption{Transition graph of the rules of $\mathtt{TokBinV}$, where 
$\rk^i \in \{0,1,2\}$ and $b^i \in \{\mathbf{0},\mathbf{1}\}$; 
$\bar{b}^i$ denotes the bit obtained by flipping $b^i$, 
i.e., $\bar{b}^i = (b^i + 1) \bmod 2$. }
\label{fig:transition-rules}
\end{figure}

Our algorithm $\mathtt{TokBinV}$ (see Algorithm~\ref{algo:TokBinV}) is composed of seven 
rules. Note that our approach is not silent: once the system has stabilized, rules \Rtok, 
\Radd, and \Rready\ are activated infinitely often by all nodes in the network. The other 
rules are executed in the stabilizing phase only. Error detection predicates used in 
Algorithm~\ref{algo:TokBinV} are formally defined in Section~\ref{subsub:ED}.

We use the state $\top$ of variable $b_v$ to indicate that a node $v$ is in an error 
state. Consequently, rule \Rer\ detects local errors: a node $v$ that is not in error 
yet observes an error (predicate $Er(v)$) and activates this rule to transition to an 
error state. In the next step, a node already in an error state resets itself via rule 
\Rreset.

 For our algorithm to function correctly, it is necessary to clean as much 
as possible of the network when an error is detected. So, when an error is 
detected when executing rule \Rer, the predicate $Er(v)$ of its neighbors 
also becomes true, allowing the error to propagate when the neighbors also 
execute rule \Rer. The root participates actively in this propagation via 
rule \RerRoot: when an inconsistency is detected in its neighborhood, the 
root sets $b_r := \top$, which in turn triggers the error detection 
mechanisms of its neighbors. Consequently, the error propagates to the 
entire network.

Rules \Rtok, \Radd, and \Rready\ govern token circulation. The system is synchronous, 
and once the algorithm converges, node ranks remain unchanged. Therefore, when tokens 
circulate, a node $v$ must receive its token(s) from every parent $u$ satisfying 
$b_u = \bo$; this is exactly the condition captured by predicate $\mathit{TakeP}(v)$ 
(Section~\ref{subsec:token-predicates}). If $\mathit{TakeP}(v)$ holds, $v$ executes 
rule \Rtok\ to acquire a token, possibly merging its parents' tokens into one. When a 
node holds a token, its children may retrieve it only if its bit is $\bo$; otherwise, 
the token disappears. In either case, the node releases its token by executing rule 
\Radd, which also increments its $b$ variable. Note that \Radd\ sets the token variable 
to $\mathit{wait}$ rather than directly to $\mathit{false}$, so that neighboring nodes 
can identify which nodes sent a token in the previous step — a distinction exploited by 
the error-detection predicates. Finally, after releasing its token, node $v$ executes 
rule \Rready\ to signal that it is ready to receive a new token. 
Figure~\ref{fig:tok} illustrates a possible execution of these three rules.

If instead $\mathit{TakeO}(v)$ holds — meaning tokens arrive but not from exactly the 
parent set — $v$ detects a structural inconsistency and executes rule \RerRank\ to 
declare itself in an error state.

Finally, rule \Rjoin\ concerns reset nodes, which do not hold a rank. When predicate 
$\mathit{TakeR}(v)$ holds, $v$ is offered a consistent token wave from a single rank 
level; it then joins the spanning structure by executing \Rjoin, which sets its rank to 
$\rk_t(v)$ and acquires a token.

Figure~\ref{fig:transition-rules} summarizes all possible state
transitions of a node $v$ under the rules of $\mathtt{TokBinV}$,
expressed in terms of the triplet $(\rk_v, t_v, b_v)$.
\begin{figure}
    \centering
    % Sous-figure 1
    \begin{subfigure}{0.05\textwidth}
        \centering
       \begin{tikzpicture}
    \node[circle, draw, minimum size=1cm, fill=lightgray] (v) at (0,0) {\footnotesize$\mathit{true},\bo$};
    \node[below of=v, circle, node distance=0.3cm, fill=black] {};
    \node[right of=v, node distance=0.8cm] {$v$};
    \node[below left of=v, node distance=1cm] {\footnotesize\Radd};
    \node[circle, draw, minimum size=1cm, fill=lightgray] (u) at (0,-2) {\footnotesize$\mathit{false},\bo$};
    \node[right of=u, node distance=0.8cm] {$u$};
    \node[below left of=u, node distance=1cm] {\footnotesize \Rtok};
    \node[circle, draw, minimum size=1cm] (w) at (0,-4) {\footnotesize$\mathit{false},\bo$};
    \node[right of=w, node distance=0.8cm] {$w$};
    \node[circle, draw, minimum size=1cm] (x) at (0,-6) {\footnotesize$\mathit{false},\bz$}; 
    \node[right of=x, node distance=0.8cm] {$x$};
     \node[circle, draw, minimum size=1cm] (y) at (0,-8) {\footnotesize$\mathit{false},\bo$}; 
    \node[right of=y, node distance=0.8cm] {$y$};
    \draw[->] (u) to (v);
    \draw[->] (w) to (u);
    \draw[->] (x) to (w); 
    \draw[->] (y) to (x); 
\end{tikzpicture}

        \caption{}
    \end{subfigure}
    \hspace{1cm}
    % Sous-figure 2
     \begin{subfigure}{0.05\textwidth}
        \centering
       \begin{tikzpicture}
    \node[circle, draw, minimum size=1cm, fill=lightgray] (v) at (0,0) {\footnotesize$\mathit{wait},\bz$};
    \node[right of=v, node distance=0.8cm] {$v$};
    \node[below left of=v, node distance=1cm] {\footnotesize \Rready};
    \node[circle, draw, minimum size=1cm, fill=lightgray] (u) at (0,-2) {\footnotesize$\mathit{true},\bo$};
    \node[right of=u, node distance=0.8cm] {$u$};
    \node[below of=u, circle, node distance=0.3cm, fill=black] {};
    \node[below left of=u, node distance=1cm] {\footnotesize \Radd};
    \node[circle, draw, minimum size=1cm,fill=lightgray] (w) at (0,-4) {\footnotesize$\mathit{false},\bo$};
    \node[right of=w, node distance=0.8cm] {$w$};
    \node[below left of=w, node distance=1cm] {\footnotesize \Rtok};
    \node[circle, draw, minimum size=1cm] (x) at (0,-6) {\footnotesize$\mathit{false},\bz$}; 
    \node[right of=x, node distance=0.8cm] {$x$};
    \node[circle, draw, minimum size=1cm] (y) at (0,-8) {\footnotesize$\mathit{false},\bo$}; 
    \node[right of=y, node distance=0.8cm] {$y$};
    \draw[->] (u) to (v);
    \draw[->] (w) to (u);
    \draw[->] (x) to (w); 
    \draw[->] (y) to (x);
        \end{tikzpicture}
        \caption{}
    \end{subfigure}
    \hspace{1cm}
    % Sous-figure 3
     \begin{subfigure}{0.05\textwidth}
        \centering
        \begin{tikzpicture}
             \node[circle, draw, minimum size=1cm] (v) at (0,0) {\footnotesize$\mathit{false},\bz$};
    \node[right of=v, node distance=0.8cm] {$v$};
    \node[circle, draw, minimum size=1cm, fill=lightgray] (u) at (0,-2) {\footnotesize$\mathit{wait},\bz$};
    \node[right of=u, node distance=0.8cm] {$u$};
    \node[below left of=u, node distance=1cm] {\footnotesize\Rready};
    \node[circle, draw, minimum size=1cm,fill=lightgray] (w) at (0,-4) {\footnotesize$\mathit{true},\bo$};
    \node[right of=w, node distance=0.8cm] {$w$};
     \node[below of=w, circle, node distance=0.3cm, fill=black] {};
    \node[below left of=w, node distance=1cm] {\footnotesize\Radd};
    \node[circle, draw, minimum size=1cm,fill=lightgray] (x) at (0,-6) {\footnotesize$\mathit{false},\bz$}; 
    \node[right of=x, node distance=0.8cm] {$x$};
    \node[below left of=x, node distance=1cm] {\footnotesize \Rtok};
    \node[circle, draw, minimum size=1cm] (y) at (0,-8) {\footnotesize$\mathit{false},\bo$}; 
    \node[right of=y, node distance=0.8cm] {$y$};
    \draw[->] (u) to (v);
    \draw[->] (w) to (u);
    \draw[->] (x) to (w); 
    \draw[->] (y) to (x);
        \end{tikzpicture}
        \caption{}
    \end{subfigure}
     \hspace{1cm}
    % Sous-figure 4
     \begin{subfigure}{0.05\textwidth}
        \centering
        \begin{tikzpicture}
    \node[circle, draw, minimum size=1cm] (v) at (0,0) {\footnotesize$\mathit{false},\bz$};
    \node[right of=v, node distance=0.8cm] {$v$};
    \node[circle, draw, minimum size=1cm] (u) at (0,-2) {\footnotesize$\mathit{false},\bz$};
    \node[right of=u, node distance=0.8cm] {$u$};
    \node[circle, draw, minimum size=1cm,fill=lightgray] (w) at (0,-4) {\footnotesize$\mathit{wait},\bz$};
    \node[right of=w, node distance=0.8cm] {$w$};
    
    \node[below left of=w, node distance=1cm] {\footnotesize\Rready};
    \node[circle, draw, minimum size=1cm,fill=lightgray] (x) at (0,-6) {\footnotesize$\mathit{true},\bz$}; 
    \node[right of=x, node distance=0.8cm] {$x$};
    \node[below left of=x, node distance=1cm] {\footnotesize\Radd};
     \node[below of=x, circle, node distance=0.3cm, fill=black] {};
    \node[circle, draw, minimum size=1cm] (y) at (0,-8) {\footnotesize$\mathit{false},\bo$}; 
    \node[right of=y, node distance=0.8cm] {$y$};
    \draw[->] (u) to (v);
    \draw[->] (w) to (u);
    \draw[->] (x) to (w); 
    \draw[->] (y) to (x);
        \end{tikzpicture}
        \caption{}
    \end{subfigure}
     \hspace{1cm}
    % Sous-figure 5
     \begin{subfigure}{0.05\textwidth}
        \centering
        \begin{tikzpicture}
    \node[circle, draw, minimum size=1cm] (v) at (0,0) {\footnotesize$\mathit{false},\bz$};
    \node[right of=v, node distance=0.8cm] {$v$};
    \node[circle, draw, minimum size=1cm] (u) at (0,-2) {\footnotesize$\mathit{false},\bz$};
    \node[right of=u, node distance=0.8cm] {$u$};
    \node[circle, draw, minimum size=1cm] (w) at (0,-4) {\footnotesize$\mathit{false},\bz$};
    \node[right of=w, node distance=0.8cm] {$w$};
    \node[circle, draw, minimum size=1cm,fill=lightgray] (x) at (0,-6) {\footnotesize$\mathit{wait},\bo$}; 
    \node[right of=x, node distance=0.8cm] {$x$};
    \node[below left of=x, node distance=1cm] {\footnotesize \Rready};
    \node[circle, draw, minimum size=1cm] (y) at (0,-8) {\footnotesize$\mathit{false},\bo$}; 
    \node[right of=y, node distance=0.8cm] {$y$};
    \draw[->] (u) to (v);
    \draw[->] (w) to (u);
    \draw[->] (x) to (w); 
    \draw[->] (y) to (x);
        \end{tikzpicture}
        \caption{}
    \end{subfigure}
    \caption{The gray nodes in each subfigure represent activatable nodes. (a)~Node $v$ possesses and offers a token. Node $u$ acquires the token by applying rule \Rtok, and node $v$ increments its bit by applying \Radd. (b)~Node $v$ applies rule \Rready\/ to indicate that it is ready to receive a new token. Node $u$ holds the token; it releases the token and increments its bit by applying \Radd. Simultaneously, node $w$ executes rule \Rtok\/ to acquire the token. (c)~Node $u$ applies \Rready; node $w$ holds the token and applies \Radd; node $x$ acquires the token via \Rtok. (d)~Node $x$ holds the token and applies \Radd; however, node $y$ does not acquire it because $b_x = \bz$. (e)~Node $x$ applies \Rready.}
    \label{fig:tok}
\end{figure}

\subsection{Error Detection Predicates}
\label{subsub:ED}

Inconsistencies between node variables can arise from an arbitrary initial configuration.
The following predicates, defined for all nodes $v \in V \setminus \{r\}$, allow the
algorithm to detect and propagate such errors. We organize them into three groups:
internal state inconsistencies, token-flow violations, and synchronization violations.

\paragraph*{Internal state inconsistencies.}

Predicate $Er_{var}(v)$ fires when a node holds $\rk_v = \bot$ but its other variables
do not match the reset state — the node claims to be reset without being so:
\begin{equation}
    Er_{var}(v) \equiv \rk_v = \bot \wedge \neg \mathit{Reset}(v)
    \label{eq:erRst}
\end{equation}

Predicate $Er_{Root}(v)$ detects a topological impossibility: a neighbor of $v$ holds
rank $0^*$ (the root's rank), yet $v$ has more than one parent, which would imply the
existence of multiple roots:
\begin{equation}
    Er_{Root}(v) \equiv \bigl(\exists u \in \Neig(v) \mid \rk_u = 0^*\bigr) \wedge |P(v)| > 1
\end{equation}

Predicate $Er_{noP}(v)$ detects a ranked node with no parent — an impossibility in any
valid $BFS^{\equiv 3}$ structure except at the root:
\begin{equation}
    Er_{noP}(v) \equiv \rk_v \neq \bot \wedge P(v) = \emptyset
\end{equation}

\paragraph*{Token-flow violations}

Predicate $Er_{tp}(v)$ fires when a child of $v$ holds a token but $v$ has not
just sent one (i.e., $v$ is not in the state $t_v = \mathit{wait} \wedge b_v = \mathbf{0}$
that immediately follows sending). In a synchronous execution, a child can only
hold a token if its parent sent one in the previous step. The condition
$\neg Er_{prg}(v)$ excludes nodes that are already receiving an error signal
from a neighbor (this propagation predicate is defined in
Section~\ref{subsub:ED}):
\begin{equation}
    Er_{tp}(v) \equiv \neg Er_{prg}(v) \wedge b_v \neq \top
    \wedge \bigl(\exists u \in \mathcal{C}(v) \mid t_u = \mathit{true}\bigr)
    \wedge \neg(t_v = \mathit{wait} \wedge b_v = \mathbf{0})
    \label{eq:ertp}
\end{equation}

Predicate $Er_{wait}(v)$ detects a node that has just released a token (state
$t_v = \mathit{wait} \wedge b_v = \bz$) while at least one of its neighbors
is a reset node. After a token emission, all unranked neighbors should have joined
the structure via \Rjoin; a remaining reset neighbor signals an inconsistency:
\begin{equation}
    Er_{wait}(v) \equiv t_v = \mathit{wait} \wedge b_v = \mathbf{0}
    \wedge \bigl(\exists u \in \Neig(v) \mid \mathit{Reset}(u)\bigr)
\label{eq:erwait}
\end{equation}

\paragraph*{Synchronization violations}

Since the network is synchronous and we construct a $BFS^{\equiv 3}$, all nodes at the
same distance from the root must progress through token-passing steps simultaneously.
This imposes that all parents of $v$ share the same state ($Er_P$), all children of $v$
share the same state ($Er_C$), and all siblings of $v$ share the same state as $v$
($ Er_S$). Additionally, a ranked node with children must not have any reset neighbor
($ Er_{Re}$), since in a correct execution all unranked neighbors join the structure when
a token is offered:
\begin{align}
    Er_P(v) &\equiv \exists u, w \in \Pa(v) \mid b_u \neq b_w \vee t_u \neq t_w
    \label{eq:Er_P}\\
    Er_C(v) &\equiv \exists u, w \in \mathcal{C}(v) \mid b_u \neq b_w \vee t_u \neq t_w
    \label{eq:ErC}\\
     Er_S(v) &\equiv \exists u \in \mathcal{S}(v) \mid b_u \neq b_v \vee t_u \neq t_v
    \label{eq:Er_S}\\
     Er_{Re}(v) &\equiv \rk_v \neq \bot \wedge C(v) \neq \emptyset
    \wedge t_v \neq \mathit{true}
    \wedge \bigl(\exists u \in \Neig(v) \mid \rk_u = \bot\bigr)
    \label{eq:Rtnd}
\end{align}
These four sub-predicates are combined into $Er_N(v)$, guarded to avoid interference
with $\mathit{TakeO}$ and the error propagation mechanism:
\begin{equation}
    Er_N(v) \equiv \neg Er_{prg}(v) \wedge b_v \neq \top \wedge \neg \mathit{TakeO}(v)
    \wedge \bigl(Er_P(v) \vee Er_C(v) \vee  Er_S(v) \vee  Er_{Re}(v)\bigr)
    \label{eq:erN}
\end{equation}

Predicate $Er_\pi(v)$ applies to reset nodes: if two neighbors of a reset node share
the same rank but differ in their state $(b, t)$, the token wave is inconsistent and
$v$ would receive tokens from only one of them, corrupting the join procedure:
\begin{equation}
    Er_\pi(v) \equiv \mathit{Reset}(v) \wedge
    \bigl(\exists u, w \in \Neig(v) \mid \rk_u = \rk_w
    \wedge (b_u \neq b_w \vee t_u \neq t_w)\bigr)
\end{equation}

\paragraph*{Error propagation.}

When a node $v$ sets $b_v := \top$ to signal an error, any neighbor that observes
$b_u = \top$ must itself enter the error state. This cascading mechanism, captured by
$Er_{prg}(v)$, ensures that errors spread across the network and trigger a global
cleanup:
\begin{equation}
    Er_{prg}(v) \equiv \exists u \in \Neig(v) \mid b_u = \top
    \label{eq:erprg}
\end{equation}

\paragraph*{Global error predicate.}

A node $v$ is considered to be in error if any of the above conditions hold:
\begin{equation}
    Er(v) \equiv Er_{Root}(v) \vee Er_{tp}(v) \vee Er_N(v) \vee Er_{noP}(v)
    \vee Er_{prg}(v) \vee Er_\pi(v) \vee Er_{wait}(v)
    \label{eq:er}
\end{equation}

Recall that $Er_{var}(v)$ is excluded from $Er(v)$: since $\rk_v = \bot$,
node $v$ has no family relationships, and thus no structural error to
propagate to its neighborhood — unlike all other predicates in $Er(v)$,
which rely on parent, child, or sibling relationships.

\paragraph*{Summary.}
Table~\ref{tab:errors} groups the error predicates by category and indicates the
node type each applies to.

\begin{table}[h]
\centering
\renewcommand{\arraystretch}{1.4}
\begin{tabular}{l l l}
\hline
\textbf{Predicate} & \textbf{Category} & \textbf{Detected violation} \\
\hline
$Er_{var}(v)$  & Internal state   & Unranked node with inconsistent variables \\
$Er_{Root}(v)$ & Internal state   & Multiple parents while a neighbor is the root $0^*$ \\
$Er_{noP}(v)$    & Internal state   & Ranked node with no parent \\
\hline
$Er_{tp}(v)$   & Token flow       & Child holds token without parent having sent one \\
$Er_{wait}(v)$ & Token flow       & Token just sent but a reset neighbor remains \\
\hline
$Er_N(v)$      & Synchronization  & Parents, children, or siblings not in identical state \\
$Er_\pi(v)$    & Synchronization  & Reset node sees inconsistent neighbors of same rank \\
\hline
$Er_{prg}(v)$  & Propagation      & A neighbor is in error state ($b_u = \top$) \\
\hline
\end{tabular}
\caption{Error detection predicates of $\mathtt{TokBinV}$, grouped by category.}
\label{tab:errors}
\end{table}

%================
\section{Correctness}
\begin{theorem}
\label{thm:main}
In a semi-uniform model where $r$ is the distinguished node, and every 
other node is anonymous, our synchronous deterministic self-stabilizing 
algorithm~$\mathtt{TokBin}$ constructs a $BFS^{\equiv 3}$  rooted in $r$ using 
$O(1)$ bits of memory per node and stabilizes in $O(2^{\varepsilon})$ steps, where $\varepsilon$ 
denotes the eccentricity of $r$.
\end{theorem}

\subsection{Space Complexity}

\begin{lemma}
\label{lemma:space}
Algorithm~$\mathtt{TokBin}$ uses $O(1)$ bits of memory per 
node.
\end{lemma}

\begin{proof}
Every node $v \in V$ — including the root $r$ — maintains exactly three 
constant variables: $\rk_v$, $t_v$, and $b_v$, each taking values in a constant 
domain: $\rk_v \in \{0, 1, 2, \bot\}$, $t_v \in \{\mathit{true}, 
\mathit{wait}, \mathit{false}\}$, and $b_v \in \{\bz, \bo, \top\}$. 
The number of states per node is therefore constant, at most 
$4 \times 3 \times 3 = 36$, requiring at most 
$\lceil \log_2 36 \rceil = 6$ bits per node, independently of any 
global network parameter.
\end{proof}

% Chaîne d'attracteurs avec le diagramme :
% Γ  ▷  Γ_ē  ▷  Γ_tr  ▷  Γ_cl  ▷  Γ*
%    O(1)  O(n)        O(2^ε)  O(2^ε)
% + mention que tout repose sur Γ_r établi d'abord

\subsection{Definitions and Notations}

Recall that $\Gamma$ denotes the set of all possible configurations of 
the system (see Section~\ref{sec:model}). Since all non-root nodes are 
anonymous, they carry no identifiers. However, for clarity of exposition, 
we assign identifiers to nodes  — they are used solely to name nodes in the exposition and carry no 
weight in either the algorithm or the proofs themselves.

For a variable $x_v$ of node $v$, we write $x_v(\gamma)$ to denote its 
value in configuration $\gamma$. We denote by $\mathcal{T}(\gamma)$ the 
set of nodes currently holding a token:
\begin{equation}
    \mathcal{T}(\gamma) = \{v \in V \mid t_v(\gamma) = \mathit{true}\}
\end{equation}
We partition $\mathcal{T}(\gamma)$ into two sets according to the value 
of the bit variable:
\begin{align}
    \mathcal{T}^{\bz}(\gamma) &= \{v \in V \mid t_v(\gamma) = 
        \mathit{true} \wedge b_v(\gamma) = \bz\} \\
    \mathcal{T}^{\bo}(\gamma) &= \{v \in V \mid t_v(\gamma) = 
        \mathit{true} \wedge b_v(\gamma) = \bo\}
    \label{eq:T1}
\end{align}
Intuitively, nodes in $\mathcal{T}^{\bo}(\gamma)$ hold a forwarding 
token — one that can be passed to their children — while nodes in 
$\mathcal{T}^{\bz}(\gamma)$ hold a suppressed token that will 
disappear after one step.
%-----
\subsection{Overview of Correctness.}

The correctness proof proceeds through a chain of attractors illustrated 
in Figure~\ref{fig:attractors}. Starting from any configuration 
$\gamma \in \Gamma$, the system traverses the following sets in order, 
each reachable from the previous in a bounded number of steps:
\begin{itemize}
    \item $\Gamma_{\bar{e}}$ — no obvious error remains 
          (Lemma~\ref{lem:errorFree}, $O(1)$ steps);
    \item $\Gamma_{\bar{it}}$ — all initial tokens have disappeared 
          (Lemma~\ref{lem:CleanBadToken}, $O(n)$ steps);
    \item $\Gamma_{cl}$ — all residual errors have been eliminated 
          (Lemma~\ref{lem:Clean}, $O(2^\varepsilon)$ steps);
    \item $\Gamma^*$ — all configurations are legal 
          (Lemma~\ref{lem:BFS}, $O(2^\varepsilon)$ steps), 
          and $\Gamma^*$ is closed (Lemma~\ref{lem:clos}).
\end{itemize}
The dominant term is $O(2^\varepsilon)$, giving the overall stabilization 
time stated in Theorem~\ref{thm:main}.
\begin{figure}[h]
\centering
\begin{tikzpicture}[
    node distance=1.4cm,
    box/.style={
        rectangle, rounded corners=6pt,
        minimum width=1.6cm, minimum height=1cm,
        draw, thick, align=center,
        font=\small
    },
    graybox/.style={box, fill=gray!15, draw=gray!50},
    tealbox/.style={box, fill=teal!15, draw=teal!50},
    purplebox/.style={box, fill=violet!15, draw=violet!50},
    amberbox/.style={box, fill=orange!15, draw=orange!50},
    arr/.style={->, thick, >=stealth},
    lbl/.style={font=\scriptsize, above=2pt}
]

%% Nodes
\node[graybox]  (G)   {$\Gamma$};
\node[tealbox,  right=of G]   (Ge)  {$\Gamma_{\bar{e}}$\\\scriptsize\textit{no obvious}\\\scriptsize\textit{errors}};
\node[purplebox, right=of Ge]  (Gt)  {$\Gamma_{\bar{it}}$\\\scriptsize\textit{no initial}\\\scriptsize\textit{tokens}};
\node[purplebox, right=of Gt]  (Gc)  {$\Gamma_{cl}$\\\scriptsize\textit{no residual}\\\scriptsize\textit{errors}};
\node[amberbox, right=of Gc]  (Gs)  {$\Gamma^*$\\\scriptsize\textit{legal}};

%% Arrows with step counts
\draw[arr] (G)  -- node[lbl]{$O(1)$}              (Ge);
\draw[arr] (Ge) -- node[lbl]{$O(n)$}              (Gt);
\draw[arr] (Gt) -- node[lbl]{$O(2^\varepsilon)$}  (Gc);
\draw[arr] (Gc) -- node[lbl]{$O(2^\varepsilon)$}  (Gs);

%% Closure arc on Γ*
\draw[->, thick, >=stealth, dashed]
    (Gs.north) to[out=120, in=60, looseness=4]
    node[above, font=\scriptsize]{closed (Lem.~\ref{lem:clos})}
    (Gs.north);

%% Lemma labels below each arrow
\draw[dashed]
    ([yshift=-0.4cm]G.south east) --
    ([yshift=-0.4cm]Ge.south west)
    node[midway, below, font=\tiny]{Lem.~\ref{lem:errorFree}};
\draw[dashed]
    ([yshift=-0.4cm]Ge.south east) --
    ([yshift=-0.4cm]Gt.south west)
    node[midway, below, font=\tiny]{Lem.~\ref{lem:CleanBadToken}};
\draw[dashed]
    ([yshift=-0.4cm]Gt.south east) --
    ([yshift=-0.4cm]Gc.south west)
    node[midway, below, font=\tiny]{Lem.~\ref{lem:Clean}};
\draw[dashed]
    ([yshift=-0.4cm]Gc.south east) --
    ([yshift=-0.4cm]Gs.south west)
    node[midway, below, font=\tiny]{Lem.~\ref{lem:BFS}};

%% Total complexity box
\node[font=\scriptsize, align=center,
      below=1.4cm of Gt]
    {Total: $O(1) + O(n) + O(2^\varepsilon) = O(2^\varepsilon)$ steps};

\end{tikzpicture}
\caption{Chain of attractors in the correctness proof of $\mathtt{TokBin}$.
Each arrow $\Gamma_A \triangleright \Gamma_B$ represents convergence
from $\Gamma_A$ to $\Gamma_B$ in the indicated number of steps.
The self-loop on $\Gamma^*$ indicates closure (Lemma~\ref{lem:clos}).}
\label{fig:attractors}
\end{figure}

\medskip

Before and throughout this chain, Lemma~\ref{lem:root} plays a 
continuous role: the root $r$ may temporarily deviate from its normal 
operating state $(\rk_r = 0^*, t_r = \mathit{true}, b_r = \bo)$ 
whenever it detects an inconsistency in its neighborhood — for 
instance, when an error propagates from one $r$-component while 
another remains unaffected. Each such deviation is resolved within 
$O(1)$ steps, after which the root resumes token emission. All 
subsequent lemmas therefore assume that $\mathit{Legal}_r(\gamma)$ 
holds, with the understanding that any temporary deviation of the root 
introduces at most an $O(1)$ additive overhead that does not affect 
the asymptotic bounds of the attractor chain. 
%------
\subsection{Root Stabilization ($\Gamma_r$)}
\label{sec:root}

Although $\mathtt{TokBinR}$ allows the root $r$ to temporarily set 
$b_r := \top$ when it detects an inconsistency in its neighborhood, 
$r$ always recovers its normal operating state within $O(1)$ steps. 
This section establishes this fact as a prerequisite for all subsequent 
lemmas, which implicitly rely on $r$ continuously emitting tokens.

We define $\mathit{Legal}_r(\gamma)$ as the predicate capturing the 
normal operating state of $r$ and its neighborhood:
\begin{equation}
    \mathit{Legal}_r(\gamma) \equiv \mathcal{N}_{Synch}(r) \wedge b_r = 1
    \label{eq:legalr_root}
\end{equation}
We define $\Gamma_r = \{\gamma \in \Gamma \mid \mathit{Legal}_r(\gamma)\}$ 
as the set of configurations where the root is in its normal operating 
state.

\begin{lemma}
\label{lem:root}
$\Gamma \triangleright \Gamma_r$ in $O(1)$ steps.
\end{lemma}

\begin{proof}
Let $\gamma_0$ be any configuration such that 
$\neg\mathit{Legal}_r(\gamma_0)$. We define the potential function 
$\lambda : \Gamma \times \{r\} \to \mathbb{N}$ as follows:
\[
\lambda(\gamma, r) =
\begin{cases}
4 & \text{if } \neg\mathit{Legal}_r(\gamma) \wedge b_r \neq \top \\
3 & \text{if } \neg\mathit{Legal}_r(\gamma) \wedge b_r = \top 
      \wedge \neg\mathcal{N}_{Er}(r) \\
2 & \text{if } \neg\mathit{Legal}_r(\gamma) \wedge b_r = \top 
      \wedge \mathcal{N}_{Er}(r) \\
1 & \text{if } \neg\mathit{Legal}_r(\gamma) \wedge b_r = 1 
      \wedge \mathcal{N}_{Reset}(r) \\
0 & \text{otherwise}
\end{cases}
\]
We show that $\lambda$ strictly decreases at each step.

\medskip
\noindent\textbf{Case $\lambda(\gamma_0, r) = 4$.}
Since $b_r \neq \top$, the guard of \Rgo\/  is false. The neighborhood 
is neither synchronized ($\neg\mathcal{N}_{Synch}(r)$), nor fully reset 
($\neg\mathcal{N}_{Reset}(r)$), nor fully in error 
($\neg\mathcal{N}_{Er}(r)$), so the guard of \RerRoot\/ holds. The root 
executes \RerRoot, setting $b_r(\gamma_1) := \top$. Hence 
$\lambda(\gamma_1, r) = 3$.

\medskip
\noindent\textbf{Case $\lambda(\gamma_1, r) = 3$.}
We have $b_r = \top$ but $\neg\mathcal{N}_{Er}(r)$, meaning at least 
one neighbor $v$ of $r$ satisfies $b_v \neq \top$. Since $b_r = \top$, 
predicate $Er_{prg}(\gamma_1, v)$ holds for every such $v$, so every 
such neighbor executes \Rer\/ synchronously, setting 
$b_v(\gamma_2) := \top$. Since the scheduler is synchronous, all such 
neighbors act simultaneously. Hence $\mathcal{N}_{Er}(r)$ holds in 
$\gamma_2$, so $\lambda(\gamma_2, r) = 2$.

\medskip
\noindent\textbf{Case $\lambda(\gamma_2, r) = 2$.}
We have $b_r = \top$ and $\mathcal{N}_{Er}(r)$, so the guard of \Rgo\/  
is satisfied. The root executes \Rgo, setting $b_r(\gamma_3) := 1$. 
Simultaneously, every neighbor $v$ of $r$ has $b_v = \top$, so it 
satisfies the guard of \Rreset\/ and executes it, setting 
$\rk_v(\gamma_3) := \bot$, $t_v(\gamma_3) := \mathit{false}$, and 
$b_v(\gamma_3) := 0$. Hence $\mathcal{N}_{Reset}(r)$ holds in $\gamma_3$, 
so $\lambda(\gamma_3, r) = 1$.

\medskip
\noindent\textbf{Case $\lambda(\gamma_3, r) = 1$.}
We have $b_r = 1$ and $\mathcal{N}_{Reset}(r)$, meaning every neighbor 
$v$ satisfies $\rk_v = \bot$, $t_v = \mathit{false}$, $b_v = 0$. Since 
$t_r = \mathit{true}$ and $b_r = 1$, we have $\mathit{Tok}(v) = \{r\}$ 
for every neighbor $v$. The predicate $\mathit{TakeR}(v)$ is satisfied: 
all senders in $\mathit{Tok}(v)$ hold rank $0^*$, and all other neighbors 
of $v$ are reset nodes with rank $\bot \neq 0^*$. Therefore every 
neighbor $v$ executes \Rjoin\/, setting 
$\rk_v(\gamma_4) := \rk_t(v) = (0^* + 1) \bmod 3 = 1$ and 
$b_v(\gamma_4) := 1$. Hence $\mathit{Legal}_r(\gamma_4)$ holds, so 
$\lambda(\gamma_4, r) = 0$.

\medskip
In every case $\lambda$ decreases strictly and reaches $0$ in at most 
$4$ steps. Hence $\Gamma \triangleright \Gamma_r$ in $O(1)$ steps.
\end{proof}

\begin{remark}
\label{rem:root}
All subsequent lemmas implicitly assume that $\mathit{Legal}_r(\gamma)$ 
holds. By Lemma~\ref{lem:root}, the system reaches $\Gamma_r$ within 
$O(1)$ steps from any initial configuration. Since this overhead is 
absorbed into the overall step counts, all subsequent complexity bounds 
remain valid.
\end{remark}
%----------------------------
\subsection{Obvious Errors ($\Gamma_{\bar{e}}$)}
\label{subsec:obvious}

We define obvious errors as the inconsistent states of nodes that are 
detectable in the first configuration $\gamma_0$ after transient faults 
stop occurring. While other errors may emerge as a consequence of these 
transient faults in subsequent configurations after $\gamma_0$, we 
address such inconsistencies later.

Let $\phi:\Gamma\times V \rightarrow \mathbb{N}$ be the following 
function:
$$\phi(\gamma,v)=
\left\{
  \begin{array}{ll}
   1 & \text{if } Er_{var}(\gamma,v) \vee Er_{tp}(\gamma,v) 
       \vee Er_N(\gamma,v)\\
   0 & \text{otherwise}
  \end{array}
\right.
$$
Let $\Phi:\Gamma \rightarrow \mathbb{N}$ be the following function:
$$\Phi(\gamma)=\sum_{v\in V\setminus\{r\}}\phi(\gamma,v)$$
Finally, we define the set of configurations $\Gamma_{\bar{e}}$ as 
follows: $\Gamma_{\bar{e}}=\{\gamma\in \Gamma:\Phi(\gamma)=0\}$.

\medskip
The following lemma is stated separately as it is also invoked in the 
proof of closure of $\Gamma^*$ (Lemma~\ref{lem:clos}).

\begin{lemma}
\label{lem:add}
For every configuration $\gamma_i \in \Gamma_{\bar{e}}$, if a node $v$ 
executes rule \Radd, \Rready, or \Rtok, then the resulting configuration 
$\gamma_{i+1}$ remains in $\Gamma_{\bar{e}}$.
\end{lemma}

\begin{proof}
We consider a configuration $\gamma_0 \in \Gamma_{\bar{e}}$, which means 
that $v$ and its neighbors do not have any obvious errors.

If $v$ has no parent, then \Radd\ is not activatable. Otherwise, if $v$ 
holds a token, every parent $u$ of $v$ satisfies $t_u(\gamma_0) = 
\mathit{wait}$ and $b_u(\gamma_0) = \bz$; otherwise, $v$ would satisfy 
$Er_{tp}(\gamma_0, v) = \mathit{true}$ (see predicate~\ref{eq:ertp}). 
Consequently, in a synchronous system, every parent of $v$ applies 
\Rready\ in $\gamma_0$.

Every sibling $u$ of $v$ is in $\mathcal{T}(\gamma_0)$; otherwise, $v$ 
would satisfy $Er_{N}(\gamma_0,v) = \mathit{true}$. Moreover, every 
sibling $u$ of $v$ satisfies $b_u(\gamma_0) = b_v(\gamma_0)$; otherwise, 
$Er_{N}(\gamma_0, v) = \mathit{true}$ (see predicate~\ref{eq:erN}). 
Consequently, in a synchronous system, every sibling $u$ of $v$ and $v$ 
itself apply \Radd\ in $\gamma_0$, and since \Radd\ increments $b$ 
identically for all of them, we have $b_u(\gamma_1) = b_v(\gamma_1)$, 
so $Er_{N}(\gamma_1, v)$ remains $\mathit{false}$.

For every child $u$ of $v \in \mathcal{T}^{\bo}(\gamma_0)$, we have 
$t_u(\gamma_0) = \mathit{false}$; otherwise, $u$ would satisfy 
$Er_{tp}(\gamma_0,u) = \mathit{true}$. Thus, every child $u$ of 
$v \in \mathcal{T}^{\bo}(\gamma_0)$ satisfies predicate 
$\mathit{TakeP}(u)$ (see predicate~\ref{eq:TakeP}), and synchronously 
applies rule \Rtok. Conversely, every child $u$ of 
$v \in \mathcal{T}^{\bz}(\gamma_0)$ does not satisfy 
$\mathit{TakeP}(u)$ and therefore remains inactive.

In conclusion, in one step, all children of nodes in 
$\mathcal{T}^{\bo}(\gamma_0)$ are included in $\mathcal{T}(\gamma_1)$, 
and all nodes in $\mathcal{T}^{\bo}(\gamma_0)$ are not in 
$\mathcal{T}(\gamma_1)$. Figure~\ref{fig:tok} illustrates the 
circulation of tokens.

Finally, note that $v$ cannot execute \Radd\ if at least one neighbor 
$u$ has \mbox{$b_u(\gamma_i) = \top$}. Furthermore, if a neighbor $u$ of $v$ 
has $b_u(\gamma_{i+1}) = \top$, then $Er_{prg}(\gamma_{i+1}, v) = 
\mathit{true}$ (see predicate~\ref{eq:erprg}), which implies 
$Er_{tp}(\gamma_{i+1}, v) = \mathit{false}$ and $Er_{N}(\gamma_{i+1}, 
v) = \mathit{false}$, so $\gamma_{i+1}$ remains in~$\Gamma_{\bar{e}}$.
\end{proof}

\begin{lemma}
\label{lem:errorFree}
$\Gamma \triangleright \Gamma_{\bar{e}}$ in $O(1)$ steps, and 
$\Gamma_{\bar{e}}$ is closed.
\end{lemma}

\begin{proof}
\textbf{Convergence.} Let $\gamma_0 \notin \Gamma_{\bar{e}}$. Let us consider 
a node $v \in V \setminus \{r\}$ such that \mbox{$Er_{var}(\gamma_0,v) = 
\mathit{true}$}, \mbox{$Er_{tp}(\gamma_0,v) = \mathit{true}$}, or 
$Er_{N}(\gamma_0,v) = \mathit{true}$ (see 
equations~\ref{eq:erRst},~\ref{eq:ertp}, and~\ref{eq:erN}). 
Consequently, $v$ executes \Rer, updating its variables to 
$t_v(\gamma_1) = \mathit{wait}$ and $b_v(\gamma_1) = \top$. Since 
these three predicates each require $b_v \neq \top$ to hold, we have 
\mbox{$\phi(\gamma_0, v) = 1 > \phi(\gamma_1, v) = 0$}. Since the system is 
synchronous, every such node behaves identically in $\gamma_0$, hence 
\mbox{$\Phi(\gamma_1) = 0$}, that is, $\gamma_1 \in \Gamma_{\bar{e}}$.

Regarding the root $r$: since $\phi$ sums over $V \setminus \{r\}$ 
only, the root does not contribute to $\Phi$. However, if 
$\neg\mathit{Legal}_r(\gamma_0)$, the root may execute \RerRoot\/ 
in $\gamma_0$, setting $b_r(\gamma_1) := \top$. By 
Lemma~\ref{lem:root}, the root reaches $\Gamma_r$ within $O(1)$ steps 
independently of the non-root nodes, and this overhead is absorbed into 
the overall $O(1)$ convergence bound.

\medskip
\textbf{Closure.} We now show that for every rule applicable in
$\gamma_i \in \Gamma_{\bar{e}}$, the resulting configuration
$\gamma_{i+1}$ remains in $\Gamma_{\bar{e}}$.
We first establish a key invariant shared by most cases. If a rule
sets $b_v(\gamma_{i+1}) := \top$ or $\rk_v(\gamma_{i+1}) := \bot$,
then every neighbor $u$ of $v$ satisfies
$Er_{prg}(\gamma_{i+1}, u) = \mathit{true}$
(see Equation~\ref{eq:erprg}), which immediately implies
$Er_{tp}(\gamma_{i+1}, u) = \mathit{false}$ and
$Er_N(\gamma_{i+1}, u) = \mathit{false}$
(see Equations~\ref{eq:ertp} and~\ref{eq:erN}).
This invariant applies to rules \Rreset\ ($\rk_v := \bot$),
\Rer, \RerRank, and \RerRoot\ ($b_v := \top$);
the case analysis below focuses on the local effect at $v$ only.

\begin{itemize}

\item \textbf{Rule \Rreset.} Sets $\rk_v := \bot$,
$t_v := \mathit{false}$, and $b_v := \bz$, so
$Er_{var}(\gamma_{i+1}, v) = \mathit{false}$
(see Equation~\ref{eq:erRst}).
The invariant handles all neighbors.

\item \textbf{Rule \Rer.} Sets $b_v := \top$, so
$Er_N(\gamma_{i+1}, v) = \mathit{false}$.
The invariant handles all neighbors.

\item \textbf{Rule \RerRank.} Sets $b_v := \top$, so
$Er_N(\gamma_{i+1}, v) = \mathit{false}$.
The invariant handles all neighbors.

\item \textbf{Rule \Rjoin.} Does not set $b_v := \top$.
However, predicate $\mathit{TakeR}(v)$ guarantees that all
token-offering neighbors share the same state, and all other
neighbors are reset nodes. Thus $v$ and its reset neighbors
execute \Rjoin\ simultaneously while every parent $u$ executes
\Radd, so $Er_N(\gamma_{i+1}, v) = \mathit{false}$ and
$Er_N(\gamma_{i+1}, w) = \mathit{false}$ for all
$w \in \mathcal{N}(v)$.

\item \textbf{Rules \RerRoot\ and \Rgo.} Executed solely by $r$;
no variable of $v \in V \setminus \{r\}$ is modified.
If \RerRoot\ sets $b_r := \top$, the invariant applies to
all neighbors of $r$, so $\phi(\gamma_{i+1}, v) = 0$.
By Lemma~\ref{lem:root}, this situation is transient and
$r$ returns to $\Gamma_r$ within $O(1)$ steps.

\item \textbf{Rules \Rtok, \Radd, and \Rready.}
See Lemma~\ref{lem:add}.

\end{itemize}
\end{proof}
%--------
\subsection{Legal Configurations: Definitions and Properties}
\label{sec:legal}

We consider a graph $G$ whose root is $r$. In a legal configuration, 
each node $v$ must have a rank $\rk_v(\gamma)$ equal to $d^r_v \bmod 3$, 
captured by predicate $Legal^d(\gamma,v)$:
\begin{equation}
    Legal^d(\gamma,v) \equiv \rk_v = d^r_v \bmod 3
    \label{eq:legald}
\end{equation}

In a legal configuration, the rank of nodes must remain stable. In our 
algorithm, ranks are modified either when a node detects an error, or 
when it receives one or more tokens from at least one node that does not 
belong to its parent set. Therefore, to ensure that ranks remain 
unchanged, it is necessary to guarantee that every node $v$ only receives 
tokens from $\Pa(v)$. In other words, all paths between $r$ and $v$ in a 
legal configuration must remain consistent: for every shortest path $\pi$ 
between $r$ and $v$, every node $u$ at the same distance from $r$ in 
$\pi$ must have identical values for variables $b_u$ and $t_u$.

More formally, let $\Pi(\gamma, r, v)$ denote the set of shortest paths 
from $r$ to $v$ (defined by the parent-child relationships) in 
configuration $\gamma$. Consider a path $\pi \in \Pi(\gamma, r, v)$ of 
the form:
$$\pi(\gamma, r, v) = r, w_0, w_1, \dots, w_i, v$$
Let $b(\gamma,\pi)$ be the binary word obtained by concatenating the 
variables $b$ of all nodes along $\pi$:
$$b(\gamma,\pi) = b_{r}(\gamma), b_{w_0}(\gamma), \dots, 
b_{w_i}(\gamma), b_v(\gamma)$$
Note that in legal configurations every node $v$ satisfies 
$b_v \neq \top$, so every entry of $b(\gamma,\pi)$ belongs to 
$\{\bz, \bo\}$ and $b(\gamma,\pi)$ is a well-defined binary word. 
Similarly, let $t(\gamma,\pi)$ be the binary word obtained by 
concatenating the variables $t$ of all nodes along $\pi$, where 
$t_v = 0$ stands for $t_v \in \{\mathit{false}, \mathit{wait}\}$ and 
$t_v = 1$ stands for $t_v = \mathit{true}$:
$$t(\gamma,\pi) = t_{r}(\gamma), t_{w_0}(\gamma), \dots, 
t_{w_i}(\gamma), t_v(\gamma)$$

The predicate $Legal^{\pi}(\gamma,v)$ captures the condition that all 
shortest paths between $r$ and $v$ carry consistent values for $t$ 
and $b$:
\begin{equation}
    Legal^{\pi}(\gamma,v) \equiv 
    \forall \pi, \pi' \in \Pi(\gamma, r, v) \mid 
    t(\gamma,\pi) = t(\gamma,\pi') \wedge b(\gamma,\pi) = b(\gamma,\pi')
    \label{eq:legalpi}
\end{equation}

A node $v \in V \setminus \{r\}$ is in a legal configuration if it 
satisfies:
\begin{equation}
    Legal(\gamma,v) \equiv Legal^d(\gamma,v) \wedge Legal^\pi(\gamma,v)
    \label{eq:legalv}
\end{equation}
By definition of the state of $r$, $Legal^d(\gamma,r) = \mathit{true}$ 
and $Legal^\pi(\gamma,r) = \mathit{true}$, so:
\begin{equation}
    Legal(\gamma,r) \equiv \mathit{true}
    \label{eq:legalr}
\end{equation}
We denote by $Legal(\gamma)$ the predicate that holds when the entire 
configuration is legal:
\begin{equation}
    Legal(\gamma) \equiv \forall v \in V \mid Legal(\gamma,v) = 
    \mathit{true}
    \label{eq:legal}
\end{equation}
The set $\Gamma^*$ is defined as the set of all configurations $\gamma$ 
such that $Legal(\gamma) = \mathit{true}$.
%---------
\subsection{Closure of Legal Configurations}
\label{sec:closure}

Recall that $G^* = G \setminus \{r\}$ denotes the graph obtained by 
removing the root. Let $G^*_1, G^*_2, \ldots, G^*_j$ denote the 
connected components of $G^*$, which we refer to as $r$-components. 
For any two distinct nodes $v, u \in V \setminus \{r\}$, $v$ and $u$ 
belong to the same $r$-component if and only if there exists a path 
between them that does not contain $r$. We denote by $G^*(v)$ the 
$r$-component of $G^*$ containing $v$.

\begin{lemma}
\label{lem:clos}
$\Gamma^*$ is closed.
\end{lemma}

\begin{proof}
We consider a single $r$-component $G^*_i$. For each configuration 
$\gamma \in \Gamma^*$, the predicate $Legal(\gamma)$ holds 
(see Equation~\ref{eq:legal}). By definition of $Legal(\gamma)$, every 
node $v$ satisfies $Legal(\gamma,v)$ and $Er(v) = \mathit{false}$. 
Hence, rule \Rer\ is not executable.

\medskip
\noindent\textbf{Closure of $Legal^d$.}
For any node $v \in V \setminus \{r\}$ in a legal configuration 
$\gamma$, it holds that $\rk_v(\gamma) \neq \bot$, ensuring \Rjoin\ 
does not apply. Furthermore, all nodes at the same distance from the 
root $r$ share the same rank $\rk$, making \RerRank\ not executable: 
since every node $v$ has $\rk_v = d^r_v \bmod 3$ and no two nodes 
at the same distance can differ in rank, the guard of \RerRank\ 
($\mathit{TakeO}(v)$, i.e., tokens arriving from non-parents) cannot 
be satisfied. Consequently, ranks remain unchanged for all nodes in 
a legal configuration. So, for all $v \in V$ and all $\gamma \in 
\Gamma^*$ with $\gamma \mapsto \gamma'$, we have 
$Legal^d(\gamma', v)$.

In addition, since every $v \in V$ satisfies $b_v(\gamma) \neq \top$ 
in $\gamma \in \Gamma^*$, rule \Rreset\ is also inapplicable, so no 
node can lose its rank.

\medskip
\noindent\textbf{Closure of $Legal^\pi$.}
Since our algorithm is not silent, token circulation continues 
indefinitely. Thanks to Lemma~\ref{lem:root}, the root always 
maintains $b_r = 1$, ensuring there is at least one token in every 
configuration. Let us denote by $\mathcal{T}(\gamma, i)$ the set of 
nodes $v$ with $t_v(\gamma) = \mathit{true}$ and $d^r_v = i$.

Since $\gamma \in \Gamma^* \subseteq \Gamma_{\bar{e}}$, every node 
$v$ and its neighbors are free of obvious errors. By 
Lemma~\ref{lem:add}, the application of rules \Radd, \Rtok, or 
\Rready\ maintains $Legal^{\pi}(\gamma', v) = \mathit{true}$. More 
precisely, in $\gamma'$:
\begin{itemize}
    \item all children of nodes in $\mathcal{T}^{\bo}(\gamma, i)$ 
    are included in $\mathcal{T}(\gamma', i+1)$: since all parents 
    at distance $i$ share the same state (by $Legal^\pi(\gamma,v)$), 
    they offer tokens simultaneously, and all their children execute 
    \Rtok\ synchronously;
    \item all nodes in $\mathcal{T}^{\bo}(\gamma, i)$ are not in 
    $\mathcal{T}(\gamma', i+1)$: they have executed \Radd, setting 
    $t_v(\gamma') = \mathit{wait}$.
\end{itemize}
Hence, for all $v \in V$ and all $\gamma \in \Gamma^*$ with 
$\gamma \mapsto \gamma'$, we have $Legal^\pi(\gamma', v)$.

\medskip
Since both $Legal^d$ and $Legal^\pi$ are preserved at every step, 
$Legal(\gamma', v)$ holds for all $v \in V$, and $\gamma' \in 
\Gamma^*$. Since the proof holds for every $r$-component $G^*_i$, 
we obtain the result for the entire graph.
\end{proof}
%--------
\subsection{Initial Tokens Lifetime ($\Gamma_{\bar{it}}$)}
\label{subsec:tokenlifetime}

This section is dedicated to proving that initial tokens at non-root 
nodes eventually disappear. We consider an execution starting in 
configuration $\gamma_0 \in \Gamma_{\bar{e}}$, so the set 
$\mathcal{T}^{\bo}(\gamma_0)$ contains all initial tokens that can be 
transmitted, except that of the root (see Equation~\ref{eq:T1}). For 
each node $v \in \mathcal{T}^{\bo}(\gamma_0)$, we define the maximal 
spanning structure within which the token from $v$ may propagate, 
denoted by $\mathcal{S}(\gamma, v)$. A node $u$ belongs to 
$\mathcal{S}(\gamma, v)$ if and only if it does not hold a token, has 
$b_u = \bo$, and all its ancestors leading to $v$ are token-free and 
have their bits set to $\bo$. These conditions are captured by the 
following predicate $\pi^1(\gamma, v, u)$:

\begin{equation}
\begin{split}
    \pi^1(\gamma,v,u) \equiv\; 
    &t_{u}(\gamma) \neq \mathit{true} \wedge b_{u}(\gamma) = \bo \;\wedge \\
    &\exists v_0, v_1, \dots, v_j, \\
    &\forall i \in \{0,\ldots,j\},\;
    v_i \in V \setminus \{r\} \wedge 
    t_{v_i}(\gamma) \neq \mathit{true} \wedge b_{v_i}(\gamma) = \bo \;\wedge \\
    &\forall i \in \{1,\ldots,j\},\;
    v_i \in \Ch(v_{i-1}) \wedge v_0 \in \Ch(v) \wedge u \in \Ch(v_j)
\end{split}
\end{equation}

Using predicate $\pi^1(\gamma, v, u)$, we define for every 
$v \in \mathcal{T}^{\bo}(\gamma)$ the spanning structure 
$\Ss(\gamma, v)$ (Figure~\ref{fig:SS1} illustrates an example):
\begin{equation}
    \Ss(\gamma, v) \equiv 
    \bigl\{u \in V \setminus \{r\} \mid \pi^1(\gamma, v, u) = 
    \mathit{true}\bigr\}
\end{equation}
Note that the eccentricity of such spanning structures is bounded by 
$n - 1$, since the root cannot be part of them.

\begin{figure}
    \centering
\begin{tikzpicture}
\node[circle, draw, minimum size=0.5cm] (v) at (0.5, 0.5) {\footnotesize$\bullet,\bo$};  
\node[above of=v,circle, node distance=0.6cm] {\footnotesize$v$};
\node[circle, draw, minimum size=0.5cm] (u1) at (-1.5, -1.5) {\footnotesize$f,\bo$}; 
\node[above of=u1,circle, node distance=0.6cm] {\footnotesize$v_1$};
\node[circle, draw, minimum size=0.5cm] (u11) at (-2.5, -2.5) {\footnotesize$f,\bo$}; 
\node[circle, draw, minimum size=0.5cm] (u12) at (-0.5, -2.5) {\footnotesize$f,\bo$}; 
\node[circle, draw, minimum size=0.5cm] (u1n1) at (-4, -6) {\footnotesize$f,\bo$};
\node[circle, draw, minimum size=0.5cm] (other) at (-6, -6) {\footnotesize$f,\bo$};
\node[below right of=other,circle, node distance=0.6cm] {\footnotesize$u_2$};
\node[above right of=u1n1,circle, node distance=0.6cm] {\footnotesize$w$};
\draw[dotted] (other) -- (u1n1);
\node[circle, draw, minimum size=0.5cm] (u1n2) at (-2.5, -5) {\footnotesize$f,\bo$}; 

\node[circle, draw, minimum size=0.5cm] (u12n1) at (0, -7.5) {\footnotesize$f,\bz$}; 
\node[circle, draw, minimum size=0.5cm] (u12n2) at (-1.5, -4) {\footnotesize$f,\bo$}; 

\node[circle, draw, minimum size=0.5cm] (nu1) at (-5, -7.5) {\footnotesize$f,\bz$}; 
\node[circle, draw, minimum size=0.5cm] (nu2) at (-2, -7.5) {\footnotesize$f,\bz$}; 
\node[below right of=nu2,circle, node distance=0.6cm] {\footnotesize$u_1$};
\node(e1) at (-6,-9){};
\node(e2) at (-4,-9){};
\node(e3) at (-2,-9){};
\draw[->] (u1) -- (v);
\draw[->] (u11) -- (u1);
\draw[->] (u12) -- (u1);
\draw[->,dashed] (u1n1) -- (u11);
\draw[->,dashed] (u1n2) -- (u11);
\draw[->,dashed] (u12n1) -- (u12);
\draw[->,dashed] (u12n2) -- (u12);
\draw[->,dashed] (nu1) -- (u1n1);
\draw[->,dashed] (nu2) -- (u12n2);
\draw[->,dashed] (nu2) -- (u12n1);
\draw[->,dashed] (nu2) -- (u1n1);
\draw[->,dashed] (e1) -- (nu1);
\draw[->,dashed] (e2) -- (nu1);
\draw[->,dashed] (e3) -- (nu2);
\draw[blue, thick, rounded corners] (-4.5, -6.5) rectangle (0.5, -0.7);
\node[blue] at (-3.5, -1.5) {$\Ss(\gamma_1, v_1)$};

\node[circle, draw, minimum size=0.5cm] (u2) at (2.5, -1.5) {\footnotesize$f,\bo$}; 
\node[above of=u2,circle, node distance=0.6cm] {\footnotesize$v_2$};
\node[circle, draw, minimum size=0.5cm] (a1) at (4.5, -3) {\footnotesize$f,\bo$}; 
\node[circle, draw, minimum size=0.5cm] (a2) at (1.5, -3) {\footnotesize$f,\bo$}; 
\node[circle, draw, minimum size=0.5cm] (a3) at (3, -3) {\footnotesize$f,\bo$};
\node[circle, draw, minimum size=0.5cm] (a4) at (4.5, -5) {\footnotesize$f,\bo$}; 
\node[circle, draw, minimum size=0.5cm] (a5) at (1.5, -4.5) {\footnotesize$f,\bo$}; 
\node[circle, draw, minimum size=0.5cm] (a6) at (3, -5) {\footnotesize$f,\bo$};
\node[circle, draw, minimum size=0.5cm] (a7) at (1.5, -7.5) {\footnotesize$f,\bz$};
\node[circle, draw, minimum size=0.5cm] (a8) at (3, -8) {\footnotesize$f,\bz$};
\node(e4) at (6,-9){};
\node(e5) at (4,-9){};
\node(e6) at (2,-9){};
\draw[->] (u2) -- (v);
\draw[->] (a1) -- (u2);
\draw[->] (a2) -- (u2);
\draw[->] (a3) -- (u2);
\draw[->,dashed] (a4) -- (a1);
\draw[->,dashed] (a4) -- (a2);
\draw[->,dashed] (a5) -- (a1);
\draw[->,dashed] (a5) -- (a3);
\draw[->,dashed] (a6) -- (a3);
\draw[->,dashed] (a7) -- (a6);
\draw[->,dashed] (a8) -- (a6);
\draw[->,dashed] (a7) -- (a5);
\draw[->,dashed] (a8) -- (a5);
\draw[->,dashed] (e4) -- (a8);
\draw[->,dashed] (e5) -- (a8);
\draw[->,dashed] (e6) -- (a7);
\draw[purple, thick, rounded corners] (6, -6.5) rectangle (1, -0.7);
\node[purple] at (5, -1.5) {$\Ss(\gamma_1, v_2)$};
\draw[ thick, rounded corners] (-5, -6.8) rectangle (6.5, 1.5);
\node[] at (5, 1) {$\Ss(\gamma_0, v)$};
\end{tikzpicture}
 \caption{$f,\bo$ means $t=\mathit{false}\wedge b=\bo$, and $f,\bz$ means $t=\mathit{false}\wedge b=\bz$}
 \label{fig:SS1}
\end{figure}

Now, let $\psi : \Gamma_{\bar{e}} \times V \rightarrow \mathbb{N}$ be 
the following function:
\begin{equation}
\psi(\gamma, v) =
    \begin{cases}
        |\Ss(\gamma,v)| & \text{if } v \in \mathcal{T}^{\bo}(\gamma) 
        \wedge \exists v' \in \mathcal{T}(\gamma_0) \mid 
        v \in \Ss(\gamma, v') \\
        0 & \text{otherwise}
    \end{cases}
\end{equation}

Let $\Psi : \Gamma_{\bar{e}} \rightarrow \mathbb{N}$ be the following 
function:
$$\Psi(\gamma) = \sum_{v \in V} \psi(\gamma, v)$$

We define the set of configurations $\Gamma_{\bar{it}}$ as follows: 
$\Gamma_{\bar{it}} = \{\gamma \in \Gamma_{\bar{e}} : \Psi(\gamma) = 0\}$. 
In other words, in a configuration $\gamma \in \Gamma_{\bar{it}}$, no token 
resulting from an initial token remains.

\begin{lemma}
\label{lem:CleanBadToken}
$\Gamma_{\bar{e}} \triangleright \Gamma_{\bar{it}}$ in $O(n)$ steps.
\end{lemma}

\begin{proof}
Let $v \in \mathcal{T}^{\bo}(\gamma_0)$ be such that 
$v \in \Ss(\gamma_0, v')$ for some $v' \in \mathcal{T}(\gamma_0)$, 
with spanning structure $\Ss(\gamma_0, v)$. In other words, $v$ 
is the root of $\Ss(\gamma_0, v)$.

\begin{itemize}

\item To satisfy $Er_{tp}(\gamma_0, v) = \mathit{false}$, every 
$u \in \Pa(v)$ must have $t_u(\gamma_0) = \mathit{wait}$ and 
$b_u(\gamma_0) = \bz$. Consequently, $u \notin \Ss(\gamma_0, v)$, 
ensuring that $\Ss(\gamma_0, v)$ contains no cycle involving $v$.

\item We now prove that if $|\Ss(\gamma_0, v)| > 1$, no node 
outside $\Ss(\gamma_0, v)$ can join $\Ss(\gamma_1, v_1)$, 
where $v_1 \in C(v)$. Since $\gamma_0 \in \Gamma_{\bar{e}}$, if a node 
$u \in \Ss(\gamma_0, v)$ has a child $w \in \Ss(\gamma_0, v)$, 
then all children of $u$ are in $\Ss(\gamma_0, v)$ (to satisfy 
$Er_{N}(\gamma_0, u) = \mathit{false}$). For a node $u$ to join 
$\Ss(\gamma_1, v_1)$, it must be a neighbor of a node already 
in $\Ss(\gamma_0, v)$.

\begin{itemize}

    \item \textbf{Case 1: $u \notin \Ss(\gamma_0, v)$ with a 
    parent $p_u \in \Ss(\gamma_0, v)$} (node $u_1$ in 
    Figure~\ref{fig:SS1}). By definition of $\Ss(\gamma_0, v)$, 
    we have $b_u(\gamma_0) = \bz$. For $b_u(\gamma_1)$ to change to 
    $\bo$, $u$ must execute \Radd, which requires $t_u(\gamma_0) = 
    \mathit{true}$. However, this would imply that $p_u$ satisfies 
    $t_{p_u}(\gamma_0) = \mathit{wait}$ and $b_{p_u}(\gamma_0) = \bz$; 
    otherwise $u$ would satisfy $Er_{tp}(\gamma_0, u) = \mathit{true}$, 
    contradicting $\gamma_0 \in \Gamma_{\bar{e}}$. But 
    $b_{p_u}(\gamma_0) = \bz$ contradicts $p_u \in \Ss(\gamma_0, v)$. 
    Hence $u$ cannot join $\Ss(\gamma_1, v_1)$.

    \item \textbf{Case 2: $u \notin \Ss(\gamma_0, v)$, not a 
    descendant of $v$, but with a neighbor $w \in \Ss(\gamma_0, v)$} 
    (node $u_2$ in Figure~\ref{fig:SS1}). Node $u$ is either a parent, 
    a sibling, or a reset node of $w$. For $u$ to become a descendant 
    of $v_1$ in $\gamma_1$, it must update its rank. Since only reset 
    nodes can adopt a new rank, parents and siblings of $w$ cannot join 
    $\Ss(\gamma_1, v_1)$. In $\Ss(\gamma_0, v)$, only 
    $v$ holds a token, which implies $w = v$. Consequently, the reset 
    node $u$ becomes a child of $v$; but since \Rjoin\ does not set 
    $t_u := \mathit{true}$, node $u$ holds no token and does not belong 
    to any $\Ss(\gamma_1, z)$ for any $z \in V$.

\end{itemize}

\item Due to predicate $Er_{prg}$ and rule \Rer, or due to rule 
\RerRank, a node $u \in \Ss(\gamma_0, v)$ may disappear from 
$\Ss(\gamma_1, v_1)$. We denote the set of such nodes by 
$D(\gamma_1, v)$.

\item To conclude, if $\mathcal{T}^{\bo}(\gamma_0) \neq \emptyset$, 
for all $v \in \mathcal{T}^{\bo}(\gamma_0)$ with 
$|\Ss(\gamma_0, v)| > 1$, we have:
$$\Ss(\gamma_0, v) = \{v\} \cup 
\bigcup_{v_1 \in C(v)} \Ss(\gamma_1, v_1) \cup D(\gamma_1, v)$$
hence $\Psi(\gamma_1) < \Psi(\gamma_0) - |\mathcal{T}^{\bo}(\gamma_0)|$.

\end{itemize}

Since initially $|\Ss(\gamma_0, v)| \leq n$, after at most $n$ 
steps no token induced by an initial token remains, so 
$\Gamma_{\bar{e}} \triangleright \Gamma_{\bar{it}}$ in $O(n)$ steps.
\end{proof}
%-----------------
\subsection{Legal Rank}
\label{subsec:legalrank}

\begin{lemma}
\label{lem:disRoot}
$\forall \gamma \in \Gamma_{\bar{it}}$, if $v \in \mathcal{T}(\gamma)$ then 
$Legal^d(\gamma, v) = \mathit{true}$.
\end{lemma}

\begin{proof}
Consider $\gamma \in \Gamma_{\bar{it}}$ and a node $v \in V$. We proceed by 
induction on the distance from $v$ to the root.

\begin{description}

\item[Base case.] The root $r$ is at distance $0$ from itself, and by 
definition $Legal^d(\gamma, r) = \mathit{true}$ for all 
$\gamma \in \Gamma$.

\item[Induction hypothesis.] Assume that every node at distance $i$ 
from $r$ that belongs to $\mathcal{T}(\gamma)$ satisfies 
$Legal^d(\gamma, v)$.

\item[Inductive step.] Consider a node $v$ at distance $i+1$ from $r$ 
such that $v \in \mathcal{T}(\gamma)$. This implies that $v$ has just 
executed \Rtok. Hence $v$ is not in error, and the following conditions 
hold: all its parents $u$ at distance $i$ from $r$ have 
$t_u(\gamma) = \mathit{wait}$ and satisfy $Legal^d(\gamma, u)$ by the 
induction hypothesis; all its siblings share the same state as $v$; 
and if $v$ has children, they all share the same state 
(since $\gamma \in \Gamma_{\bar{e}}$). As a result, $v$ has not 
modified its rank, so:
$$\rk_v(\gamma) = (\rk_u(\gamma) + 1) \bmod 3 = 
((i \bmod 3) + 1) \bmod 3 = (i+1) \bmod 3$$
Hence $v$ satisfies $Legal^d(\gamma, v)$.

\item[Conclusion.] By the induction principle, the property holds for 
every node in $G$.

\end{description}
\end{proof}
%----
\subsection{Token Dissemination}
\label{sec:tokDis}

Given the possible values $\{\mathit{false}, \mathit{true}, \mathit{wait}\}$ 
for variable $t$, the root's neighbors receive a token (\textit{i.e.}, 
$t_v = \mathit{true}$) every $3$ steps. To ensure that a node $v$ located 
at distance $d^r_v$ from the root receives a token, the root must transmit 
$2^{d^r_v}$ tokens, assuming no token is already present and that every 
ancestor $u$ of $v$ has not yet registered the presence of a token 
(\textit{i.e.}, $b_u = \bz$). Specifically, every neighbor $w$ of the root 
must transmit $2^{d^r_v - 1} - 1$ tokens, which must subsequently pass 
through $d^r_v - 1$ intermediate nodes between $w$ and $v$. As a result, 
a node $v$ at distance $d^r_v$ from the root receives a token every 
$\mathcal{R}(v)$ steps, where:
\begin{equation}
    \mathcal{R}(v) = (2^{d^r_v - 1} - 1) \times 3 + d^r_v - 1
    \label{eq:Rv}
\end{equation}
Note that for every $v \in V \setminus \{r\}$, $d^r_v \leq \varepsilon$, 
where $\varepsilon$ denotes the eccentricity of $r$. By 
Lemma~\ref{lem:root}, the root continuously maintains $b_r = 1$ after 
$O(1)$ steps, which guarantees the regularity of the token emission assumed 
in $\mathcal{R}(v)$. If an ancestor $u$ of $v$ already holds a token or 
has registered the presence of a token (\textit{i.e.}, $b_u = \bo$), $v$ 
may receive the token faster.

\begin{figure}[htbp]
\begin{center}
\begin{tikzpicture}[scale=0.6]

\node[] (t) at (-2.5,0) {};
\node[below=0.1cm of t] (tb) {$b(\gamma,\pi)=$};
\node[below=0.6cm of t] (tt) {$t(\gamma,\pi)=$};
\node[below=1.2cm of t] (tbb) {$\Bb^2(\gamma,\pi)=$};
\node[circle,draw,minimum size=0.5cm] (u) at (0,0) {\textbullet};
\node[above=0.1cm of u] (un) { $r$};
%\node[below=1.6cm of u] () { $2^1$};
%\node[below=0.1cm of u] (ub) {$1$};
%\node[below=0.6cm of u] (ut) {$1$};
%\node[below=1.2cm of u] (ubb) {$0$};
\node[circle,draw,,minimum size=0.5cm] (w0) at (1.5,0) {\textbullet};
\node[above=0.1cm of w0] (w0t) { $w_0$};
\node[below=1.6cm of w0] () { $2^1$};

\node[below=0.1cm of w0] (w0b) {$1$};
\node[below=0.6cm of w0] (w0t) {$1$};
\node[below=1.2cm of w0] (w0bb) {$0$};

\draw[<-] (u) -- (w0);
\node[circle,draw,minimum size=0.5cm] (w1) at (3,0) {};
\node[above=0.1cm of w1] (w1t) { $w_1$};
\node[below=1.6cm of w1] () { $2^2$};
\node[below=0.1cm of w1] (w1b) {$1$};
\node[below=0.6cm of w1] (w1t) {$0$};
\node[below=1.2cm of w1] (w1bb) {$0$};

\draw[<-] (w0) -- (w1);
\node[circle,draw,minimum size=0.5cm] (w2) at (4.5,0) {};
\node[above=0.1cm of w2] (w2t) { $w_2$};
\node[below=1.6cm of w2] () { $2^3$};
\node[below=0.1cm of w2] (w2b) {$0$};
\node[below=0.6cm of w2] (w2t) {$0$};
\node[below=1.2cm of w2] (w2bb) {$1$};

\draw[<-] (w1) -- (w2);
\node[circle,draw,minimum size=0.5cm] (w3) at (6,0) {\textbullet};
\node[above=0.1cm of w3] (w3t) { $w_3$};
\node[below=1.6cm of w3] () { $2^4$};
\node[below=0.1cm of w3] (w3b) {$0$};
\node[below=0.6cm of w3] (w3t) {$1$};
\node[below=1.2cm of w3] (w3bb) {$1$};
\draw[<-] (w2) -- (w3);
\node[circle,draw,minimum size=0.5cm] (w4) at (7.5,0) {};
\node[above=0.1cm of w4] (w4t) { $w_4$};
\node[below=1.6cm of w4] () { $2^5$};
\node[below=0.1cm of w4] (w4b) {$1$};
\node[below=0.6cm of w4] (w4t) {$0$};
\node[below=1.2cm of w4] (w4bb) {$1$};
\draw[<-] (w3) -- (w4);
\node[circle,draw,minimum size=0.5cm] (w5) at (9,0) {};
\node[above=0.1cm of w5] (w5t) { $w_5$};
\node[below=1.6cm of w5] () { $2^6$};
\node[below=0.1cm of w5] (w5b) {$0$};
\node[below=0.6cm of w5] (w5t) {$0$};
\node[below=1.2cm of w5] (w5bb) {$0$};
\draw[<-] (w4) -- (w5);
\node[circle,draw,minimum size=0.5cm] (w6) at (10.5,0) {\textbullet};
\node[above=0.1cm of w6] (w6t) { $w_6$};
\node[below=1.6cm of w6] () { $2^7$};
\node[below=0.1cm of w6] (w6b) {$1$};

\node[below=0.6cm of w6] (w6t) {$1$};
\node[below=1.2cm of w6] (w6bb) {$0$};
\draw[<-] (w5) -- (w6);
\node[circle,draw,minimum size=0.5cm] (w7) at (12,0) {};
\node[above=0.1cm of w7] (w7t) { $w_7$};
\node[below=1.6cm of w7] () { $2^8$};
\node[below=0.1cm of w7] (w7b) {$0$};
\node[below=0.6cm of w7] (w6t) {$0$};
\node[below=1.2cm of w7] (w6bb) {$1$};

\draw[<-] (w6) -- (w7);
\node[circle,draw,minimum size=0.5cm] (v) at (13.5,0) {};
\node[above=0.1cm of v] (vt) { $v$};
\node[below=1.6cm of v] () { $2^{9}$};
\node[below=0.1cm of v] (vb) {$0$};
\node[below=0.6cm of v] (vt) {$0$};
\node[below=1.2cm of v] (vbb) {$0$};
\draw[<-] (w7) -- (v);

\node[minimum size=0.5cm] (b) at (15,0) {};
%\node[below=1.6cm of b] () { $2^{11}$};
%\node[below=0.1cm of b] () {$0$};
%\node[below=0.6cm of b] () {$0$};
%\node[below=1.2cm of b] () {$0$};

\draw (-1,-2.3)-- (14.5,-2.3);
\node (ad) at (-0.5, -1.5){$\oplus$};
\end{tikzpicture}
\vspace{0.5cm}

\caption{Dotted nodes denote nodes with a token. We have $d^r_v=9;\mathcal{R}(v)=(2^{9-1}-1)\times 3 +9-1=773;\Bb^2(\gamma,v)=10011100;\Bb^{10}(\gamma,v)=156;\mathcal{R}^c(\gamma,v)=773-156=617$ steps. }
\label{fig:Bin}
\end{center}
\end{figure}

We define $\mathcal{B}^2(\gamma, \pi)$ as the binary word obtained by the 
binary addition (denoted by $\oplus$) of $b(\gamma, \pi)$ and 
$t(\gamma, v)$, as defined in Section~\ref{sec:legal}. Then, 
$\mathcal{B}^{10}(\gamma, \pi)$ denotes the decimal representation of 
$\mathcal{B}^2(\gamma, \pi)$, assuming the highest weight for bits 
further from the root (see Figure~\ref{fig:Bin}). Given a path 
$\pi \in \Pi(\gamma, r, v)$, we denote by $\hat{\pi}(\gamma, v)$ the 
path carrying the nearest token to $v$:
\begin{equation}
    \hat{\pi}(\gamma, v) = \max_{\pi \in \Pi(\gamma, r, v)} 
    \mathcal{B}^{10}(\gamma, \pi)
    \label{eq:hatpi}
\end{equation}
Finally, $\mathcal{R}^c(\gamma, v)$ represents the number of steps 
required for node $v$ to receive a token, assuming $v$ never satisfies 
$Er(v)$:
\begin{equation}
    \mathcal{R}^c(\gamma, v) = \mathcal{R}(v) - \hat{\pi}(\gamma, v)
    \label{eq:Rc}
\end{equation}
Figure~\ref{fig:Bin} illustrates an example of this computation.
%---------
\subsection{Residual Error Removal ($\Gamma_{cl}$)}
\label{subsec:residual}

This section focuses on the elimination of residual errors, including 
incorrect ranks. The elimination of incorrect ranks occurs either through 
the application of rule \Rer\ due to predicate $Er_{prg}$, or through 
the application of rule \RerRank. Consider a configuration 
$\gamma \in \Gamma_{\bar{it}}$; the set $\mathcal{E}_k(\gamma)$ is the set of 
nodes that would execute rule \RerRank\ should a token be proposed to 
them:
\begin{equation}
    \mathcal{E}_k(\gamma) = \{v \in V \mid \rk_v(\gamma) \neq \bot 
    \wedge Legal^d(\gamma,v) = \mathit{false}\}
\end{equation}
Note that cycles contain at least one node in $\mathcal{E}_k(\gamma)$, 
while a spanning structure that contains no cycle may have no node in 
$\mathcal{E}_k(\gamma)$. The set $\mathcal{E}_\pi(\gamma)$ captures the 
nodes with $Legal^\pi(\gamma,v) = \mathit{false}$:
\begin{equation}
    \mathcal{E}_\pi(\gamma) = \{v \in V \mid \rk_v(\gamma) \neq \bot 
    \wedge Legal^\pi(\gamma,v) = \mathit{false}\}
\end{equation}

\subsubsection{Cleaner Wave}
\label{subsub:clean}

The set $\mathcal{E}(\gamma)$ is the set of error nodes in configuration 
$\gamma$:
\begin{equation}
    \mathcal{E}(\gamma) = \{v \in V \mid b_v(\gamma) = \top\}
\end{equation}

The predicate $\pi^G(\gamma, v, u)$ holds if there exists a path in $G$ 
between a node $v \in \mathcal{E}_k(\gamma)$ and a node 
$u \in \mathcal{E}(\gamma)$ that avoids all error nodes. Formally, for 
every $v \in \mathcal{E}_k(\gamma)$ and $u \in \mathcal{E}(\gamma)$:
\begin{equation}
\begin{split}
    \pi^G(\gamma,v,u) \equiv 
    \bigl(\exists v_0, v_1, \dots, v_j \mid 
    \forall i \in \{0,\dots,j\},\; v_i \in V \setminus \{r\} 
    \wedge v_i \notin \mathcal{E}(\gamma)\bigr) \;\wedge \\
    (\forall i < j \mid v_i \in \Neig(v_{i+1})) 
    \wedge v \in \Neig(v_0) \wedge v_j \in \Neig(u)
\end{split}
\end{equation}
If $\pi^G(\gamma, v, u) = \mathit{true}$ for at least one node 
$u \in \mathcal{E}(\gamma)$, then node $v \in \mathcal{E}_k(\gamma)$ 
eventually executes rule \Rer\ due to $Er_{prg}$. The nodes responsible 
for this behavior are referred to as \emph{cleaner nodes}, denoted by 
$Cl$, and are formally defined as follows:
\begin{equation}
    Cl(\gamma, v) = \{u \in G^*(v) \mid \pi^G(\gamma, v, u) = 
    \mathit{true}\}
\end{equation}
More precisely, node $v \in \mathcal{E}_k(\gamma)$ eventually executes 
\Rer\ due to its nearest cleaner node. This is captured by the function 
$E(\gamma, v)$, which returns the distance to the nearest cleaner node:
\begin{equation}
    E(\gamma, v) =
    \begin{cases}
        \min_{u \in Cl(\gamma,v)} d^u_v & \text{if } 
        Cl(\gamma,v) \neq \emptyset \\
        \infty & \text{otherwise}
    \end{cases}
\end{equation}
The function $E(\gamma, v)$ computes how long it takes for a cleaner 
wave to reach $v$.

\subsubsection{Legal Rank Wave}
\label{subsub:goodrank}

In a configuration $\gamma \in \Gamma_{\bar{it}}$, it may happen that a node 
$v$ holding a token originating from the root no longer has a path to 
the root (\textit{e.g.}, because some nodes on the path are now reset 
nodes). However, reset nodes can be treated as nodes $u$ on the path 
between $r$ and $v$ with $t_u(\gamma) = \mathit{false}$ and 
$b_u(\gamma) = \bz$. Only the acceptance of a new token by the neighbors 
of the root can reduce $\mathcal{R}^c(\gamma, v)$ by one, for any 
$v \in V$; this happens every $3$ steps. To capture this, we introduce 
the function $\widehat{\mathcal{R}^c}(\gamma, v)$ as follows:
\begin{equation}
    \widehat{\mathcal{R}^c}(\gamma, v) =
    \begin{cases}
        \mathcal{R}^c(\gamma,v) & 
            \text{if } \forall u \in \Neig(r),\; t_u(\gamma) = \mathit{true}\\
        \mathcal{R}^c(\gamma,v) - 0.3 & 
            \text{if } \forall u \in \Neig(r),\; t_u(\gamma) = \mathit{wait}\\
        \mathcal{R}^c(\gamma,v) - 0.6 & 
            \text{if } \forall u \in \Neig(r),\; t_u(\gamma) = \mathit{false}
    \end{cases}
    \label{eq:hatRc}
\end{equation}
Once $\mathcal{R}^c(\gamma, v)$ reaches zero, the token must descend to 
the relevant node. Recall that for a node $v$ to receive a token from 
node $w$, all nodes $u$ between the token and $v$ must have $b_u = \bo$. 
If $v$ may receive a token from its closest ancestor $w$ in configuration 
$\gamma$, $\beta(\gamma, v)$ returns the number of steps required to 
do so:
\begin{equation}
    \beta(\gamma, v) =
    \begin{cases}
        d^u_v & \text{if } \exists u \in \mathcal{T}^{\bo}(\gamma) 
                \mid v \in \Ss(\gamma, u) \\
        \widehat{\mathcal{R}^c}(\gamma,v) & \text{otherwise}
    \end{cases}
\end{equation}

\subsubsection{Legal Path Wave}
\label{subsub:legalpath}

Due to an erroneous initial configuration $\gamma \in \Gamma_{\bar{it}}$, 
there may exist a node $v$ that does not satisfy $Legal^\pi(\gamma, v)$, 
but is neither a reset node nor an error node (see \textit{e.g.}, node 
$v$ in Figure~\ref{fig:badPath}(a)). If $Legal^\pi(\gamma, v) = 
\mathit{false}$ because at least one ancestor of $v$ is a false root $u$ 
(\textit{i.e.}, there is no path between the root and $v$ passing through 
$u$), then $u$ satisfies $\mathit{Er_{noP}}(\gamma, u)$ and executes \Rer. 
Node $v$ is then cleaned by the propagation of the cleaner wave (see 
Section~\ref{subsub:clean}). If $r$ is an ancestor of $v$ regardless of 
the path taken, and $Legal^\pi(\gamma, v) = \mathit{false}$, this implies 
that the set of paths is not synchronized with respect to variables $t$ 
and $b$ (see Equation~\ref{eq:legalpi}). In this case, either $v$ is 
cleaned by a cleaner wave, or $v$ receives a token from only a subset of 
its parents. In the latter case, $v$ satisfies $\mathit{TakeO}(v)$ 
(see Equation~\ref{eq:TakeO}) and executes rule \RerRank, which marks 
it as erroneous.

Consider a configuration $\gamma_i \in \Gamma_{\bar{it}}$ where 
$\forall v \in V,\; Er(\gamma_i, v) = \mathit{false}$. Suppose there 
exist two paths of equal length, $\pi_1$ and $\pi_2$, starting from the 
root and ending respectively at $x_1$ and $x_2$, such that 
$Legal(\gamma_i, x_1) = \mathit{true}$ and 
$Legal(\gamma_i, x_2) = \mathit{true}$, but with either 
$b(\gamma, \pi_1) \neq b(\gamma, \pi_2)$ or 
$t(\gamma, \pi_1) \neq t(\gamma, \pi_2)$. If $x_1$ and $x_2$ are not 
in the same $r$-component, then for all $j > i$, 
$\forall v \in V,\; Er(\gamma_j, v) = \mathit{false}$. On the other 
hand, if $x_1$ and $x_2$ belong to the same $r$-component, there exists 
a path between them that does not pass through $r$. Let $v$ denote the 
first reset node shared by the extensions of paths $\pi_1$ and $\pi_2$ 
(see Figure~\ref{fig:badPath}(b)). If $v$ joins either $\pi_1$ or 
$\pi_2$ in configuration $\gamma_j$ with $j > i$, then 
$Legal^\pi(\gamma_j, v) = \mathit{false}$. To prevent this scenario, 
when $v$ is offered a token from only one of $\pi_1$ or $\pi_2$, it 
enters an error state via $Er_\pi(\gamma_j, v)$, thereby triggering the 
cleaning of both paths.

\paragraph{Role of $Er_{wait}$.}
A subtler synchronization issue arises when a reset node $v$ has two 
neighbors $u_1$ and $u_2$ sharing the same rank, but where only $u_1$ 
offers a token while $u_2$ does not. Since $v$ is a reset node, it 
cannot compare the states of $u_1$ and $u_2$ directly — it simply 
observes an inconsistent token wave and enters an error state via 
$Er_\pi(v)$, blocking the join. However, from $v$'s perspective alone, 
this situation is indistinguishable from a correct wave where $u_2$ is 
simply not yet ready. The predicate $Er_{wait}$ (see 
Equation~\ref{eq:erwait}) resolves this ambiguity from the \emph{sender's} 
side: node $u_1$, after emitting its token (entering state 
$t_{u_1} = \mathit{wait} \wedge b_{u_1} = \bz$), observes that $v$ is 
still in the reset state. This means that $v$ did not execute \Rjoin, 
which implies that $\mathit{TakeR}(v)$ was not satisfied — and therefore 
that $u_2$ did not offer a token at the same rank as $u_1$. Node $u_1$ 
thus detects a path synchronization error and executes \Rer, setting 
$b_{u_1} := \top$. This triggers $Er_{prg}$ at $v$ and $u_2$, 
initiating a cleaner wave that removes the inconsistent paths. The number 
of steps required for this process is captured by $\beta(\gamma, v)$ (see 
Section~\ref{subsub:goodrank}).

\begin{figure}
\hspace{2cm}
        \begin{subfigure}{0.10\textwidth}
       \begin{tikzpicture}
    \node[circle, draw, minimum size=1.1cm, double] (v) at (1,0) {\footnotesize$0,\mathit{t},\bo$};
    \node[below of=v, circle, node distance=0.3cm, fill=black] {};
    \node[right of=v, node distance=0.8cm] {$r$};
    \node[circle, draw, minimum size=1cm, fill=lightgray] (u) at (0,-2) {\footnotesize$1,\mathit{f},\bo$};
    \node[right of=u, node distance=0.8cm] {$u_1$};
    \node[circle, draw, minimum size=1cm, fill=lightgray] (u2) at (2,-2) {\footnotesize$1,\mathit{f},\bo$};
    \node[right of=u2, node distance=0.8cm] {$u_2$};
    \node[circle, draw, minimum size=1cm] (w) at (0,-4) {\footnotesize$2,\mathit{f},\bo$};
    \node[right of=w, node distance=0.8cm] {$w_1$};
    \node[circle, draw, minimum size=1cm] (w2) at (2,-4) {\footnotesize$2,\mathit{f},\bz$};
    \node[right of=w2, node distance=0.8cm] {$w_2$};
    \node[circle, draw, minimum size=1cm] (x) at (0,-6) {\footnotesize$0,\mathit{f},\bz$}; 
    \node[right of=x, node distance=0.8cm] {$x_1$};
    \node[circle, draw, minimum size=1cm] (x2) at (2,-6) {\footnotesize$0,\mathit{f},\bz$}; 
    \node[right of=x2, node distance=0.8cm] {$x_2$};
     \node[circle, draw, minimum size=1cm] (y) at (1,-8) {\footnotesize$1,\mathit{f},\bz$}; 
    \node[right of=y, node distance=0.8cm] {$v$};
    \draw[->] (u) to (v);
    \draw[->] (u2) to (v);
    \draw[->] (w) to (u);
    \draw[->] (w2) to (u2);
    \draw[->] (x) to (w); 
    \draw[->] (x2) to (w2); 
    \draw[->] (y) to (x2); 
    \draw[->] (y) to (x); 
\end{tikzpicture}
\caption{}
\end{subfigure}
\hspace{3cm}
% Sous-figure 2
\begin{subfigure}{0.10\textwidth}
\begin{tikzpicture}
    \node[circle, draw, minimum size=1.1cm, double] (v) at (1,0) {\footnotesize$0,\mathit{t},\bo$};
    \node[below of=v, circle, node distance=0.3cm, fill=black] {};
    \node[right of=v, node distance=0.8cm] {$r$};
    \node[circle, draw, minimum size=1cm, fill=lightgray] (u) at (0,-2) {\footnotesize$1,\mathit{f},\bo$};
    \node[right of=u, node distance=0.8cm] {$u_1$};
    \node[circle, draw, minimum size=1cm, fill=lightgray] (u2) at (2,-2) {\footnotesize$1,\mathit{f},\bo$};
    \node[right of=u2, node distance=0.8cm] {$u_2$};
    \node[circle, draw, minimum size=1cm] (w) at (0,-4) {\footnotesize$2,\mathit{f},\bo$};
    \node[right of=w, node distance=0.8cm] {$w_1$};
    \node[circle, draw, minimum size=1cm] (w2) at (2,-4) {\footnotesize$2,\mathit{f},\bz$};
    \node[right of=w2, node distance=0.8cm] {$w_2$};
    \node[circle, draw, minimum size=1cm] (x) at (0,-6) {\footnotesize$0,\mathit{f},\bz$}; 
    \node[right of=x, node distance=0.8cm] {$x_1$};
    \node[circle, draw, minimum size=1cm] (x2) at (2,-6) {\footnotesize$0,\mathit{f},\bz$}; 
    \node[right of=x2, node distance=0.8cm] {$x_2$};
     \node[circle, draw, minimum size=1cm] (y) at (1,-8) {\footnotesize$Reset$}; 
    \node[right of=y, node distance=0.8cm] {$v$};
    \draw[->] (u) to (v);
    \draw[->] (u2) to (v);
    \draw[->] (w) to (u);
    \draw[->] (w2) to (u2);
    \draw[->] (x) to (w); 
    \draw[->] (x2) to (w2); 
    \draw[-,dashed] (y) to (x2); 
    \draw[-,dashed] (y) to (x); 
\end{tikzpicture}
\caption{}
\end{subfigure}
\label{fig:notLegalPi}
        \caption{ The gray nodes are activatable. The information displayed on each node, respectively, corresponds to the rank, the value of the variable $t$, and the value of the variable $b$. (a) The node $v$ is not in error in this configuration, but it does not satisfy $Legal^{\pi}(\gamma,v)$. In fact, regarding the variable $b$, the path $u_1, w_1, x_1$ currently has the value $110$, while  $u_2, w_2, x_2$ has $100$, so $Legal^{\pi}(\gamma,v)=\mathit{false}$. This discrepancy means that $v$ will receive the next token from node $x_1$, so $v$ will receive a token from only one of its parents rather than from both. Consequently, $TakeO(v)$ will be true, and node $v$ will execute \RerRank. (b) In this case as well, $x_1$ will obtain the token before $x_2$. Consequently, node $v$ executes \Rer\/ due to $Er_\pi(v)$, and in the next step, both $x_1$ and $x_2$ enter an error state. Step by step, the two paths are deleted.  }
        \label{fig:badPath}
\end{figure}

\subsubsection{Proof of Residual Error Removal}
\label{subsub:proofclean}

Let $\rho : \Gamma_{\bar{it}} \times V \rightarrow \mathbb{N}$ be the 
following potential function:
\begin{equation}
    \rho(\gamma, v) =
    \begin{cases}
        \min\{E(\gamma,v), \beta(\gamma,v)\} \times n & 
            \text{if } v \in \mathcal{E}_k(\gamma) \vee 
            v \in \mathcal{E}_\pi(\gamma) \\
        E(\gamma,v) \times n & 
            \text{if } v \notin \mathcal{E}_k(\gamma) \wedge 
            v \notin \mathcal{E}_\pi(\gamma) \wedge 
            Cl(\gamma,v) \neq \emptyset \\
        1 & \text{if } v \in \mathcal{E}(\gamma) \\
        0 & \text{otherwise}
    \end{cases}
\end{equation}
Let $\varrho : \Gamma_{\bar{it}} \rightarrow \mathbb{N}$ be defined as:
$$\varrho(\gamma) = \sum_{v \in V} \rho(\gamma, v)$$
We define $\Gamma_{cl} = \{\gamma \in \Gamma_{\bar{it}} : \varrho(\gamma) = 0\}$.

\begin{lemma}
\label{lem:Clean}
$\Gamma_{\bar{it}} \triangleright \Gamma_{cl}$ in $O(2^\varepsilon)$ steps.
\end{lemma}

\begin{proof}
Let $\gamma_0$ be a configuration of $\Gamma_{\bar{it}}$.
%--
If $v \in \mathcal{E}_k(\gamma_0) \cup \mathcal{E}_\pi(\gamma_0)$, then
$\rho(\gamma_0, v) = \min\{E(\gamma_0,v), \beta(\gamma_0,v)\} \times n$.
We distinguish three sub-cases.

\medskip
\noindent\textbf{Sub-case $\beta(\gamma_0,v) < E(\gamma_0,v)$:}
We have $\rho(\gamma_0,v) = \beta(\gamma_0,v) \times n$.
No cleaner node can reach $v$ before a token does, so only the
change of variable $t$ by the root's neighbors reduces
$\widehat{\mathcal{R}^c}(\gamma_0, v)$ (see Equation~\ref{eq:hatRc}),
which happens at every step. Hence $\beta(\gamma_1, v) 
\beta(\gamma_0, v)$, so $\rho(\gamma_1,v) < \rho(\gamma_0,v)$.

\medskip
\noindent\textbf{Sub-case $E(\gamma_0,v) < \beta(\gamma_0,v)$:}
We have $\rho(\gamma_0,v) = E(\gamma_0,v) \times n$.
The neighbor $w$ of $u$ such that $u$ is the nearest cleaner node of
$v$ executes rule \Rer\ and becomes the nearest cleaner node of $v$.
As a consequence, $E(\gamma_1,v) = E(\gamma_0,v) - 1$, so
$\rho(\gamma_1,v) < \rho(\gamma_0,v)$.

\medskip
\noindent\textbf{Sub-case $\beta(\gamma_0,v) = E(\gamma_0,v) > 1$:}
Both mechanisms act simultaneously: $\widehat{\mathcal{R}^c}$ decreases
and the cleaner wave advances. Hence
$\min\{E(\gamma_1,v), \beta(\gamma_1,v)\} 
\min\{E(\gamma_0,v), \beta(\gamma_0,v)\}$, so
$\rho(\gamma_1,v) < \rho(\gamma_0,v)$.
%---

If $v \in \mathcal{E}_k(\gamma_0) \cup \mathcal{E}_\pi(\gamma_0)$ and 
$\beta(\gamma_0, v) = E(\gamma_0, v) = 1$ (in which case 
$\rho(\gamma_0, v) = n$), then $v$ executes either \Rer\ or \RerRank. 
In configuration $\gamma_1$, $v \in \mathcal{E}(\gamma_1)$, so 
$\rho(\gamma_1, v) = 1 < \rho(\gamma_0, v)$.

If $v \notin \mathcal{E}_k(\gamma_0) \cup \mathcal{E}_\pi(\gamma_0)$ 
and $Cl(\gamma_0, v) \neq \emptyset$, the neighbor $w$ of $u$ such that 
$u$ is the nearest cleaner node of $v$ executes rule \Rer\ and becomes 
the nearest cleaner node of $v$. As a consequence, 
$E(\gamma_1) = E(\gamma_0) - 1$, so $\rho(\gamma_1, v) < \rho(\gamma_0, v)$.

If $v$ executes \Rer\ due to $Er_{wait}(\gamma_0, v)$, then 
$b_v(\gamma_1) = \top$, so $v \in \mathcal{E}(\gamma_1)$ and 
$\rho(\gamma_1, v) = 1$. Since $Er_{wait}$ requires $b_v = \bz \neq \top$, 
we have $\rho(\gamma_0, v) > 1$, hence 
$\rho(\gamma_1, v) < \rho(\gamma_0, v)$. Moreover, every neighbor $u$ 
of $v$ with $b_u \neq \top$ now satisfies $Er_{prg}(\gamma_1, u) = 
\mathit{true}$, contributing to the propagation of the cleaner wave.

Finally, a node $v \in \mathcal{E}(\gamma_0)$ executes rule \Rreset, 
so $\rho(\gamma_1, v) = 0 < \rho(\gamma_0, v) = 1$.

In all cases, $\varrho(\gamma)$ strictly decreases by a positive bounded 
amount (the only non-integer value used is $0.3$, see 
Equation~\ref{eq:hatRc}). The strictly decreasing function $\varrho$ 
therefore reaches zero, and the number of steps is bounded by 
$2^\varepsilon$ (see Section~\ref{sec:tokDis}).
\end{proof}

\begin{corollary}
\label{cor:legalReset}
$\forall \gamma \in \Gamma_{cl}$ and $\forall v \in V$, $v$ either 
satisfies $Legal(\gamma, v)$ or $Reset(\gamma, v)$.
\end{corollary}

\begin{corollary}
\label{cor:legalpath}
$\forall \gamma \in \Gamma_{cl}$, $\forall u, v \in G^*_i$ with 
$d^r_u = d^r_v$, $\forall \pi_v \in \Pi(\gamma, r, v)$ and 
$\forall \pi_u \in \Pi(\gamma, r, u)$, we have 
$b(\gamma, \pi_u) = b(\gamma, \pi_v)$ and 
$t(\gamma, \pi_u) = t(\gamma, \pi_v)$.
\end{corollary}
%--------
\subsection{Convergence to Legal Configurations ($\Gamma^*$)}
\label{subsec:convergencelegal}

Recall that the closure of $\Gamma^*$ has already been established in 
Lemma~\ref{lem:clos}: once the system reaches a legal configuration, 
it remains legal forever. It therefore remains to prove convergence, 
\textit{i.e.}, that every execution starting from $\Gamma_{cl}$ 
eventually reaches $\Gamma^*$.

Let $\xi : \Gamma_{cl} \times V \rightarrow \mathbb{N}$ be the 
following potential function:
\begin{equation}
    \xi(\gamma, v) =
    \begin{cases}
        \beta(\gamma, v) + 1 & \text{if } Reset(\gamma, v) = 
        \mathit{true}\\
        0 & \text{otherwise}
    \end{cases}
\end{equation}
Let $\Xi : \Gamma_{cl} \rightarrow \mathbb{N}$ be defined as:
\begin{equation}
    \Xi(\gamma) = \sum_{v \in V} \xi(\gamma, v)
\end{equation}
Note that $\Xi(\gamma) = 0$ if and only if every node $v$ satisfies 
$Legal(\gamma, v) = \mathit{true}$, \textit{i.e.}, $\gamma \in \Gamma^*$.

\medskip
Recall that $\beta(\gamma, v)$ (defined in 
Section~\ref{subsub:goodrank}) measures how many steps are needed for 
the nearest forwarding token to reach $v$, where a forwarding token is 
a token held by a node $u$ with $b_u = \bo$ — recall that 
$\mathcal{T}^{\bo}(\gamma) = \{u \in V \mid t_u(\gamma) = \mathit{true} 
\wedge b_u(\gamma) = \bo\}$ (see Equation~\ref{eq:T1}) denotes the set 
of all such nodes in configuration $\gamma$. If a token from 
$\mathcal{T}^{\bo}(\gamma)$ can directly reach $v$ through its spanning 
structure, then $\beta(\gamma, v) = d^u_v$ is the distance to that 
token; otherwise $\beta(\gamma, v) = \widehat{\mathcal{R}^c}(\gamma, v)$ 
counts the steps until the root emits a token that can reach $v$. The 
potential function $\xi(\gamma, v) = \beta(\gamma, v) + 1$ thus 
represents the number of steps remaining before $v$ can join the 
spanning structure via rule \Rjoin.
\begin{lemma}
\label{lem:BFS}
$\Gamma_{cl} \triangleright \Gamma^*$ in $O(2^\varepsilon)$ steps.
\end{lemma}

\begin{proof}
By Corollary~\ref{cor:legalReset}, every configuration $\gamma_0 \in 
\Gamma_{cl}$ contains only reset nodes and nodes $v$ with 
$Legal(\gamma_0, v) = \mathit{true}$. By 
Corollary~\ref{cor:legalpath}, legal nodes within the same 
$r$-component that are at the same distance from the root share the 
same state. In a synchronous setting, such nodes apply the same rule 
simultaneously and consequently remain in synchronized states forever.

We now show that $\Xi$ strictly decreases at each step. Let $v$ be a 
reset node, so $\xi(\gamma_0, v) = \beta(\gamma_0, v) + 1 \geq 1$. 
We distinguish three cases based on the value of $\beta(\gamma_0, v)$, 
which captures how far the nearest token is from $v$.

\medskip
\noindent\textbf{Case 1: $\beta(\gamma_0, v) = 
\widehat{\mathcal{R}^c}(\gamma_0, v)$.}
No token from $\mathcal{T}^{\bo}(\gamma_0)$ can reach $v$ directly 
— $v$ is not in the spanning structure of any such token. The only 
way to reduce $\beta(\gamma_0, v)$ is the introduction of a new token 
along every path $\pi(\gamma, r, v)$ by the root's emission. This 
event occurs every $3$ steps, whenever the root's neighbors change 
their variable $t_u$, reducing $\widehat{\mathcal{R}^c}(\gamma_0, v)$ 
by $0.3$ per step (see Equation~\ref{eq:hatRc}). Hence 
$\beta(\gamma_1, v) < \beta(\gamma_0, v)$, so 
$\xi(\gamma_1, v) < \xi(\gamma_0, v)$.

\medskip
\noindent\textbf{Case 2: $\beta(\gamma_0, v) = d^u_v$ for some 
$u \in \mathcal{T}^{\bo}(\gamma_0)$ with $v \in \Ss(\gamma_0, u)$.}
There exists an ancestor $u$ of $v$ holding a forwarding token, and 
every node $w$ on the path between $u$ and $v$ has $b_w(\gamma_0) = 
\bo$. After one step, the token is accepted by every child $w$ of $u$, 
so the nearest token moves one step closer to $v$:
$$\beta(\gamma_1, v) = d^w_v = d^u_v - 1 < d^u_v = \beta(\gamma_0, v)$$
Hence $\xi(\gamma_1, v) < \xi(\gamma_0, v)$.

\medskip
\noindent\textbf{Case 3: $\beta(\gamma_0, v) = 0$.}
The nearest token is at distance $0$ from $v$, meaning $v$'s parents 
all hold a forwarding token and $v \in \Ss(\gamma_0, u)$ for 
some $u \in \mathcal{T}^{\bo}(\gamma_0)$ at distance $1$ from $v$. 
Since $\gamma_0 \in \Gamma_{cl}$, by Corollary~\ref{cor:legalpath} 
all nodes proposing the token are at the same distance from the root 
and share the same state. Predicate $\mathit{TakeR}(v)$ is therefore 
satisfied, and $v$ executes \Rjoin. After this step, $v$ joins the 
spanning structure and satisfies $Legal(\gamma_1, v)$, so 
$\xi(\gamma_1, v) = 0 < 1 = \xi(\gamma_0, v)$.

\medskip
In all three cases $\xi(\gamma_1, v) < \xi(\gamma_0, v)$, hence 
$\Xi(\gamma_1) < \Xi(\gamma_0)$. The strictly decreasing function 
$\Xi$ reaches zero in a number of steps bounded by $2^\varepsilon$ 
(see Section~\ref{sec:tokDis}).
\end{proof}
%============

\section{Concluding Remarks}

We presented the first constant space deterministic self-stabilizing BFS tree construction algorithm that operates in semi-uniform synchronous networks of arbitrary topology. Our solution~$\mathtt{TokBin}$ is independent of any global network parameter (maximum degree, diameter, number of nodes, etc.). As classical techniques such as point-to-neighbor and distance-to-the-root cannot be used in such a constrained setting, we introduced a novel infinite token stream technique to break cycles and restore BFS construction that may be of independent interest for other tasks. 
We briefly mention how our BFS construction can be used to optimally color bipartite graphs (that is, with two colors \textsc{Black} and \textsc{White}). All nodes run $\mathtt{TokBin}$ so that eventually a $BFS^{\equiv 3}$ is constructed. With respect to colors, the root has color \textsc{Black} and retains it forever, while every other node takes the opposite color of its parents whenever it executes $\mathtt{TokBin}$ (if it has parents). After stabilization of the $BFS^{\equiv 3}$ construction, all children of the root take the same color, and after one step their children take the root color, and so on, so that a 2-coloring of the bipartite graph is eventually achieved. This very simple application of $\mathtt{TokBin}$ still uses constant memory, as opposed to dedicated previous solutions to the same problem found in the literature~\cite{SS93,KK06} that require $\Omega(\log n)$ bits per node.

Two important open problems remain. 
First, our solution heavily relies on the network operation being fully synchronous (all nodes execute code in a lock-step synchronous fashion). It would be interesting to extend our token stream technique to an asynchronous setting. In this setting, we expect that some synchronization technique will be necessary, and it remains unclear whether such additional mechanism can be done in constant space. 
Second, while our algorithm uses a constant amount of memory per node, its time complexity (measured in steps) is exponential. 
This is due to the exponential decay mechanism for tokens that may require the root to send an exponential number of them to finally reach the furthest nodes in the network. One could speed up this process by adding more states to the $b$ variables: if the domain of $b$ becomes $\{\bz,\bo,\ldots,\bk,\top\}$, and token are accepted for any value that is not $\bz$ or $\top$ (assuming $b$ is incremented modulo $k$ anytime the incoming token is accepted), less tokens are destroyed, but the overall time complexity remains exponential. 
By contrast, solutions that make use of $O(\log \Delta + \log \log n)$ bits per node have polynomial time complexity~\cite{BT18}. Whether there exists a constant space and polynomial time solution to the problem is left for future research.
\bibliographystyle{plain}
\bibliography{biblio}

@article{D74,
  author    = {Edsger W. Dijkstra},
  title     = {Self-stabilizing Systems in Spite of Distributed Control},
  journal   = {Commun. {ACM}},
  volume    = {17},
  number    = {11},
  pages     = {643--644},
  year      = {1974},
  doi       = {10.1145/361179.361202}
}

@inproceedings{AKY90,
  author       = {Yehuda Afek and
                  Shay Kutten and
                  Moti Yung},
  editor       = {Jan van Leeuwen and
                  Nicola Santoro},
  title        = {Memory-Efficient Self Stabilizing Protocols for General Networks},
  booktitle    = {Distributed Algorithms, 4th International Workshop, {WDAG} '90, Bari,
                  Italy, September 24-26, 1990, Proceedings},
  series       = {Lecture Notes in Computer Science},
  volume       = {486},
  pages        = {15--28},
  publisher    = {Springer},
  year         = {1990},
  url          = {https://doi.org/10.1007/3-540-54099-7\_2},
  doi          = {10.1007/3-540-54099-7\_2},
  bibsource    = {dblp computer science bibliography, https://dblp.org}
}

@ARTICLE{CYH91,
	AUTHOR = "NS Chen and HP Yu and ST Huang",
	TITLE = "A self-stabilizing algorithm for  constructing 
	 spanning trees",
	JOURNAL = "Information Processing Letters",
	VOLUME = 39,
	YEAR = 1991,
	PAGES = "147-151"
}

@ARTICLE{DIM93,
        AUTHOR = "S Dolev and A Israeli and S Moran",
        TITLE = "Self-stabilization of dynamic systems
         assuming only read/write atomicity",
        JOURNAL = "Distributed Computing",
        VOLUME = 7,
        YEAR = 1993,
        PAGES = "3-16"
}

@article{SS93,
  author       = {Sumit Sur and
                  Pradip K. Srimani},
  title        = {A self-stabilizing algorithm for coloring bipartite graphs},
  journal      = {Inf. Sci.},
  volume       = {69},
  number       = {3},
  pages        = {219--227},
  year         = {1993},
  url          = {https://doi.org/10.1016/0020-0255(93)90121-2},
  doi          = {10.1016/0020-0255(93)90121-2},
  bibsource    = {dblp computer science bibliography, https://dblp.org}
}

@article{CD94,
  author       = {Zeev Collin and
                  Shlomi Dolev},
  title        = {Self-Stabilizing Depth-First Search},
  journal      = {Inf. Process. Lett.},
  volume       = {49},
  number       = {6},
  pages        = {297--301},
  year         = {1994},
  url          = {https://doi.org/10.1016/0020-0190(94)90103-1},
  doi          = {10.1016/0020-0190(94)90103-1},
  bibsource    = {dblp computer science bibliography, https://dblp.org}
}

@article{AG94,
  author       = {Anish Arora and
                  Mohamed G. Gouda},
  title        = {Distributed Reset},
  journal      = {{IEEE} Trans. Computers},
  volume       = {43},
  number       = {9},
  pages        = {1026--1038},
  year         = {1994},
  url          = {https://doi.org/10.1109/12.312126},
  doi          = {10.1109/12.312126},
  bibsource    = {dblp computer science bibliography, https://dblp.org}
}

@inproceedings{IL94,
  author       = {Gene Itkis and
                  Leonid A. Levin},
  title        = {Fast and Lean Self-Stabilizing Asynchronous Protocols},
  booktitle    = {35th Annual Symposium on Foundations of Computer Science, Santa Fe,
                  New Mexico, USA, 20-22 November 1994},
  pages        = {226--239},
  publisher    = {{IEEE} Computer Society},
  year         = {1994},
  url          = {https://doi.org/10.1109/SFCS.1994.365691},
  doi          = {10.1109/SFCS.1994.365691},
  bibsource    = {dblp computer science bibliography, https://dblp.org}
}

@inproceedings{AO94,
  author       = {Baruch Awerbuch and
                  Rafail Ostrovsky},
  editor       = {James H. Anderson and
                  David Peleg and
                  Elizabeth Borowsky},
  title        = {Memory-Efficient and Self-Stabilizing Network \{RESET\} (Extended
                  Abstract)},
  booktitle    = {Proceedings of the Thirteenth Annual {ACM} Symposium on Principles
                  of Distributed Computing, Los Angeles, California, USA, August 14-17,
                  1994},
  pages        = {254--263},
  publisher    = {{ACM}},
  year         = {1994},
  url          = {https://doi.org/10.1145/197917.198104},
  doi          = {10.1145/197917.198104},
  bibsource    = {dblp computer science bibliography, https://dblp.org}
}

@INPROCEEDINGS{J97,
	AUTHOR = "C Johnen",
	TITLE = "Memory-efficient self-stabilizing algorithm 
	to construct {BFS} spanning trees",
	PAGES = "125-140",
	BOOKTITLE = "Proceedings of the Third Workshop 
	 on Self-Stabilizing Systems",
	PUBLISHER = "Carleton University Press",
	YEAR = {1997}
}

@article{AB98,
  author       = {Yehuda Afek and
                  Anat Bremler{-}Barr},
  title        = {Self-Stabilizing Unidirectional Network Algorithms by Power Supply},
  journal      = {Chic. J. Theor. Comput. Sci.},
  volume       = {1998},
  year         = {1998},
  url          = {http://cjtcs.cs.uchicago.edu/articles/1998/3/contents.html},
  bibsource    = {dblp computer science bibliography, https://dblp.org}
}

@article{DGS99,
  author       = {Shlomi Dolev and
                  Mohamed G. Gouda and
                  Marco Schneider},
  title        = {Memory Requirements for Silent Stabilization},
  journal      = {Acta Informatica},
  volume       = {36},
  number       = {6},
  pages        = {447--462},
  year         = {1999},
  url          = {https://doi.org/10.1007/s002360050180},
  doi          = {10.1007/S002360050180},
  bibsource    = {dblp computer science bibliography, https://dblp.org}
}

@article{HP00,
  author       = {Ted Herman and
                  Sriram V. Pemmaraju},
  title        = {Error-detecting codes and fault-containing self-stabilization},
  journal      = {Inf. Process. Lett.},
  volume       = {73},
  number       = {1-2},
  pages        = {41--46},
  year         = {2000},
  url          = {https://doi.org/10.1016/S0020-0190(99)00164-7},
  doi          = {10.1016/S0020-0190(99)00164-7},
  bibsource    = {dblp computer science bibliography, https://dblp.org}
}

@article{DattaJPV00,
  author       = {Ajoy Kumar Datta and
                  Colette Johnen and
                  Franck Petit and
                  Vincent Villain},
  title        = {Self-stabilizing depth-first token circulation in arbitrary rooted
                  networks},
  journal      = {Distributed Comput.},
  volume       = {13},
  number       = {4},
  pages        = {207--218},
  year         = {2000},
  url          = {https://doi.org/10.1007/PL00008919},
  doi          = {10.1007/PL00008919},
  bibsource    = {dblp computer science bibliography, https://dblp.org}
}

@misc{A05,
        AUTHOR = "James Aspnes",
	TITLE = "Distributed Breadth First Search",
        NOTE = "Lecture notes of Distributed Computing, Fall 2005,
        http://pine.cs.yale.edu/pinewiki/DistributedBreadthFirstSearch",
        YEAR = 2005,
}

@inproceedings{GCH06,
  author       = {Mohamed G. Gouda and
                  Jorge Arturo Cobb and
                  Chin{-}Tser Huang},
  editor       = {Ajoy Kumar Datta and
                  Maria Gradinariu},
  title        = {Fault Masking in Tri-redundant Systems},
  booktitle    = {Stabilization, Safety, and Security of Distributed Systems, 8th International
                  Symposium, {SSS} 2006, Dallas, TX, USA, November 17-19, 2006, Proceedings},
  series       = {Lecture Notes in Computer Science},
  volume       = {4280},
  pages        = {304--313},
  publisher    = {Springer},
  year         = {2006},
  url          = {https://doi.org/10.1007/978-3-540-49823-0\_21},
  doi          = {10.1007/978-3-540-49823-0\_21},
  bibsource    = {dblp computer science bibliography, https://dblp.org}
}

@inproceedings{KK06,
  author       = {Adrian Kosowski and
                  Lukasz Kuszner},
  editor       = {Leszek Rutkowski and
                  Ryszard Tadeusiewicz and
                  Lotfi A. Zadeh and
                  Jacek M. Zurada},
  title        = {Self-stabilizing Algorithms for Graph Coloring with Improved Performance
                  Guarantees},
  booktitle    = {Artificial Intelligence and Soft Computing - {ICAISC} 2006, 8th International
                  Conference, Zakopane, Poland, June 25-29, 2006, Proceedings},
  series       = {Lecture Notes in Computer Science},
  volume       = {4029},
  pages        = {1150--1159},
  publisher    = {Springer},
  year         = {2006},
  url          = {https://doi.org/10.1007/11785231\_120},
  doi          = {10.1007/11785231\_120},
  bibsource    = {dblp computer science bibliography, https://dblp.org}
}

@article{CFIKP08,
  author       = {Reuven Cohen and
                  Pierre Fraigniaud and
                  David Ilcinkas and
                  Amos Korman and
                  David Peleg},
  title        = {Label-guided graph exploration by a finite automaton},
  journal      = {{ACM} Trans. Algorithms},
  volume       = {4},
  number       = {4},
  pages        = {42:1--42:18},
  year         = {2008},
  url          = {https://doi.org/10.1145/1383369.1383373},
  doi          = {10.1145/1383369.1383373},
  bibsource    = {dblp computer science bibliography, https://dblp.org}
}

@article{BDLP08,
  author    = {Christian Boulinier and
               Ajoy Kumar Datta and
	       Lawrence L. Larmore and
               Franck Petit},
  title     = {Time and Space Optimal Distributed {BFS} Tree Construction},
  journal   = {Information Processing Letters},
  volume    = {108},
  number    = {5},
  pages     = {273-278},
  ee        = {http://dx.doi.org/10.1016/j.ipl.2008.05.016},
  year      = {2008}
}

@inproceedings{BDPR10,
  author       = {L{\'{e}}lia Blin and
                  Shlomi Dolev and
                  Maria Gradinariu Potop{-}Butucaru and
                  Stephane Rovedakis},
  editor       = {Nancy A. Lynch and
                  Alexander A. Shvartsman},
  title        = {Fast Self-stabilizing Minimum Spanning Tree Construction - Using Compact
                  Nearest Common Ancestor Labeling Scheme},
  booktitle    = {Distributed Computing, 24th International Symposium, {DISC} 2010,
                  Cambridge, MA, USA, September 13-15, 2010. Proceedings},
  series       = {Lecture Notes in Computer Science},
  volume       = {6343},
  pages        = {480--494},
  publisher    = {Springer},
  year         = {2010},
  url          = {https://doi.org/10.1007/978-3-642-15763-9\_46},
  doi          = {10.1007/978-3-642-15763-9\_46},
  bibsource    = {dblp computer science bibliography, https://dblp.org}
}

@inproceedings{KKM11,
  author       = {Amos Korman and
                  Shay Kutten and
                  Toshimitsu Masuzawa},
  editor       = {Cyril Gavoille and
                  Pierre Fraigniaud},
  title        = {Fast and compact self stabilizing verification, computation, and fault
                  detection of an {MST}},
  booktitle    = {Proceedings of the 30th Annual {ACM} Symposium on Principles of Distributed
                  Computing, {PODC} 2011, San Jose, CA, USA, June 6-8, 2011},
  pages        = {311--320},
  publisher    = {{ACM}},
  year         = {2011},
  url          = {https://doi.org/10.1145/1993806.1993866},
  doi          = {10.1145/1993806.1993866},
  bibsource    = {dblp computer science bibliography, https://dblp.org}
}

@article{BPR13,
  author       = {L{\'{e}}lia Blin and
                  Maria Potop{-}Butucaru and
                  Stephane Rovedakis},
  title        = {A super-stabilizing log(\emph{n})log(n)-approximation algorithm for
                  dynamic Steiner trees},
  journal      = {Theor. Comput. Sci.},
  volume       = {500},
  pages        = {90--112},
  year         = {2013},
  url          = {https://doi.org/10.1016/j.tcs.2013.07.003},
  doi          = {10.1016/J.TCS.2013.07.003},
  bibsource    = {dblp computer science bibliography, https://dblp.org}
}

@article{BT18,
  author    = {L{\'{e}}lia Blin and
               S{\'{e}}bastien Tixeuil},
  title     = {Compact deterministic self-stabilizing leader election on a ring:
               the exponential advantage of being talkative},
  journal   = {Distributed Computing},
  volume    = {31},
  number    = {2},
  pages     = {139--166},
  year      = {2018},
  doi       = {10.1007/s00446-017-0294-2}
}

@article{BT20,
  author       = {L{\'{e}}lia Blin and
                  S{\'{e}}bastien Tixeuil},
  title        = {Compact self-stabilizing leader election for general networks},
  journal      = {J. Parallel Distributed Comput.},
  volume       = {144},
  pages        = {278--294},
  year         = {2020},
  url          = {https://doi.org/10.1016/j.jpdc.2020.05.019},
  doi          = {10.1016/J.JPDC.2020.05.019},
  bibsource    = {dblp computer science bibliography, https://dblp.org}
}

@article{BFB23,
  author       = {L{\'{e}}lia Blin and
                  Laurent Feuilloley and
                  Gabriel Le Bouder},
  title        = {Optimal Space Lower Bound for Deterministic Self-Stabilizing Leader
                  Election Algorithms},
  journal      = {Discret. Math. Theor. Comput. Sci.},
  volume       = {25},
  number       = {1},
  year         = {2023},
  url          = {https://doi.org/10.46298/dmtcs.9335},
  doi          = {10.46298/DMTCS.9335},
  bibsource    = {dblp computer science bibliography, https://dblp.org}
}

@article{DDJL23,
  author       = {Ajoy K. Datta and
                  St{\'{e}}phane Devismes and
                  Colette Johnen and
                  Lawrence L. Larmore},
  title        = {Analysis of a memory-efficient self-stabilizing {BFS} spanning tree
                  construction},
  journal      = {Theor. Comput. Sci.},
  volume       = {955},
  pages        = {113804},
  year         = {2023}
}

@book{D00,
  author    = {Shlomi Dolev},
  title     = {Self-Stabilization},
  publisher = {{MIT} Press},
  year      = {2000},
  isbn      = {0-262-04178-2},
}

@book{ADDP19,
  editor    = {Karine Altisen and St{\'e}phane Devismes and Swan Dubois and Franck Petit},
  title     = {Introduction to Distributed Self-Stabilizing Algorithms},
  publisher = {Morgan \& Claypool Publishers},
  series    = {Synthesis Lectures on Distributed Computing},
  year      = {2019}
}

@incollection{H22,
  author       = {Ted Herman},
  editor       = {Krzysztof R. Apt and
                  Tony Hoare},
  title        = {Origin of Self-Stabilization},
  booktitle    = {Edsger Wybe Dijkstra: His Life, Work, and Legacy},
  series       = {{ACM} Books},
  volume       = {45},
  pages        = {81--104},
  publisher    = {{ACM} / Morgan {\&} Claypool},
  year         = {2022},
  url          = {https://doi.org/10.1145/3544585.3544592},
  doi          = {10.1145/3544585.3544592},
  bibsource    = {dblp computer science bibliography, https://dblp.org}
}

\end{document}